\documentclass[num-refs]{wiley-article}

\usepackage{amsmath,amsfonts,amssymb,amscd}
\usepackage{latexsym}
\usepackage{tcolorbox}
\usepackage[all]{xy}
\usepackage{tikz}
\usetikzlibrary{arrows,shapes,positioning,automata,backgrounds,fit}

\usepackage{graphicx}
\usepackage[export]{adjustbox}

\usepackage{algorithm}
\usepackage{algorithmicx}
\usepackage[noend]{algpseudocode}
\usepackage{algpascal}

\newtheorem{claim}{Claim}

\newtheorem{definition}{Definition}

\papertype{Article}

\title{Triangle-Free Equimatchable Graphs}

\author[1]{Yasemin B\"{u}y\"{u}k\c{c}olak}
\author[2]{Didem G\"{o}z\"{u}pek}
\author[1]{Sibel \"{O}zkan}

\affil[1]{Department of Mathematics, Gebze Technical University, Kocaeli, Turkey}
\affil[2]{Department of Computer Engineering, Gebze Technical University, Kocaeli, Turkey}

\corraddress{Yasemin B\"{u}y\"{u}k\c{c}olak, Department of Mathematics, Gebze Technical University, Kocaeli, Turkey}
\corremail{y.buyukcolak@gtu.edu.tr}

\fundinginfo{This work is supported by the Scientific and Technological Research Council of Turkey (TUBITAK) under grant no. 118E799.}

\runningauthor{B\"{u}y\"{u}k\c{c}olak, G\"{o}z\"{u}pek and \"{O}zkan}

\begin{document}

\maketitle

\begin{abstract}
A graph is called equimatchable if all of its maximal matchings have the same size. Frendrup et al. \cite{Frendrup} provided a characterization of equimatchable graphs with girth at least $5$. In this paper, we extend this result by providing a complete structural characterization of equimatchable graphs with girth at least $4$, i.e., equimatchable graphs with no triangle, by identifying the equimatchable triangle-free graph families. Our characterization also extends the result given by Akbari et al. in \cite{AkbariRegular}, which proves that the only connected triangle-free equimatchable $r$-regular graphs are $C_5$, $C_7$, and $K_{r,r}$, where $r$ is a positive integer. Given a non-bipartite graph, our characterization implies a linear time recognition algorithm for triangle-free equimatchable graphs.

\end{abstract}

\section{Introduction}
A graph $G$ is \emph{equimatchable}  if every maximal matching of $G$ is a maximum matching, i.e., every maximal matching in $G$ has the same cardinality. The concept of equimatchability was introduced in 1974 independently by Gr\"{u}nbaum \cite{Grunbaum}, Lewin \cite{Lewin} and  Meng \cite{Meng}. In 1984, Lesk et al. \cite{Lesk} accomplished a formal introduction of equimatchable graphs and provided a structural characterization of equimatchable graphs via Gallai-Edmonds decomposition, yielding a polynomial-time recognition algorithm. Equimatchable graphs can also be considered as an analogue of well-covered graphs in terms of matchings. A graph is \emph{well-covered} if all of its maximal independent sets have the same size. Noting that a set of edges in a graph is a matching if and only if the corresponding set of vertices in its line graph is an independent set, a graph is equimatchable if and only if its line graph is well-covered. Since every line graph is claw-free, the polynomial-time recognition algorithm for well-covered claw-free graphs given in \cite{Tankus} determines also the equimatchability of a graph in polynomial-time. Alternatively, Demange and Ekim \cite{DemangeEkim} gave another characterization of equimatchable graphs yielding a more efficient recognition algorithm using alternating chain-based arguments in some properly constructed auxiliary graphs.

In the literature, the structure of equimatchable graphs is extensively studied by several authors. The first study is the characterization of equimatchable graphs with a perfect matching, i.e., randomly matchable graphs. In 1979, Sumner \cite{Sumner} proved that the only connected randomly matchable graphs are the complete graph $K_{2n}$ and complete bipartite graph $K_{n,n}$ for $n \geq 1$. On the other hand, the work in \cite{Lesk} provides a general structure of not only randomly matchable graphs but also equimatchable graphs without a perfect matching. A significant consequence of this characterization is that every $2$-connected equimatchable graph is either bipartite or factor-critical or $K_{2n}$, $n \geq 1$. In 1986, Favaron \cite{Favaron86} investigated equimatchable factor-critical graphs and characterized such graphs with vertex connectivity one and two. In \cite{Eiben13}, Eiben and Kotrb\v{c}\'{\i}k  proved that for a $2$-connected factor-critical equimatchable graph $G$ and a minimal matching $M$ isolating a vertex $v$ of $G$, $G \backslash (V(M) \cup \{v\})$ is a connected randomly matchable graph. More recently, Dibek et al. \cite{Dibek} showed that equimatchable graphs do not have a forbidden subgraph characterization since the equimatchability is not a hereditary property, i.e., it is not necessarily preserved by induced subgraphs. They then provided the first family of forbidden induced subgraphs of equimatchable graphs. Namely, they showed that equimatchable graphs do not contain  odd cycles of length at least nine. Furthermore, Akbari et al. \cite{AkbariClaw-Free} provided the characterization of claw-free equimatchable graphs by identifying the equimatchable claw-free graph families.

In 2010, Frendrup et al. \cite{Frendrup} gave a characterization for equimatchable graphs with girth at least five. Particularly, they showed that an equimatchable graph with girth at least five is either one of $C_5$ and $C_7$ or a member of a graph family which contains $K_2$ and all bipartite graphs with partite sets $V_1$ and $V_2$ such that all vertices in $V_1$ are stems and no vertex from $V_2$ is a stem. Recently, Akbari et al. \cite{AkbariRegular} partially characterized equimatchable regular graphs by showing that for an odd positive integer $r$, if $G$ is a connected equimatchable $r$-regular graph, then $G \in \{K_{r+1}, K_{r,r}\}$. A particular result in \cite{AkbariRegular} is that the only connected triangle-free equimatchable $r$-regular graphs are $C_5$, $C_7$ and $K_{r,r}$ where $r$ is a positive integer. In this paper, we extend both the result by Frendrup et al. \cite{Frendrup} and the result by Akbari et al. \cite{AkbariRegular} by providing a complete structural characterization of triangle-free equimatchable graphs. Our characterization yields a linear time algorithm that recognizes whether a given non-bipartite graph is equimatchable and triangle-free.

Section 2 is devoted to  basic definitions, notations, and previous results on equimatchable graphs. In Section 3, we first point out that non-factor-critical equimatchable triangle-free graphs correspond to equimatchable bipartite graphs, which have already been characterized in \cite{Lesk}. We then discuss the class of factor-critical equimatchable triangle-free graphs. In Section 4, we structurally characterize this graph class by identifying graph families in this class. In Section 5, we summarize our structural characterization for equimatchable triangle-free graphs by giving the main theorem (Theorem \ref{MainTheorem}) and provide the linear-time recognition algorithm that decides whether a given non-bipartite graph is equimatchable and triangle-free. Finally, in Section 6 we conclude the paper and present some open questions.

\section{Preliminaries}
In this section, we first give graph-theoretical definitions and notations, and then present preliminary results that will be used in the characterization of triangle-free equimatchable graphs.

All graphs in this paper are finite, simple, and undirected. For a graph $G$, $V(G)$ and $E(G)$ denote the set of vertices and edges in $G$, respectively. An edge {joining} the vertices $u$ and $v$ in $G$ will be denoted by $uv$. For a vertex $v$ in $G$ and a subset $X \subseteq V(G)$, $N(v)$ denotes the set of neighbors of $v$ in $G$, while $N(X)$ denotes the set of all vertices adjacent to at least one vertex of $X$ in $G$. A vertex of degree one is called a \emph{leaf} and a vertex adjacent to a leaf is called a \emph{stem}. The \emph{order} of $G$ is denoted by $|V(G)|$; besides, $G$ is called \emph{odd} (resp. \emph{even}) if $|V(G)|$ is odd (resp. even).
For a graph $G$ and $U \subseteq V(G)$, the subgraph {induced by} $U$ is denoted by $G[U]$. The \emph{difference} $G\backslash H$ of two graphs $G$ and $H$ is defined as the subgraph induced by the difference of their vertex sets, i.e., $G \backslash H = G[V(G)\backslash V(H)]$. For a graph $G$ and a vertex $v$ of $G$, the subgraph induced by $V(G) - v$ is denoted by $G - v$ for the sake of brevity.
The {path}, {{cycle}}, and {complete graph} on $n$ vertices are denoted by $P_n$, $C_n$, and $K_n$, respectively, while the {complete bipartite graph} with bipartition of  sizes $n$ and $m$ is denoted by $K_{n,m}$. The graph $C_3$ (or equivalently $K_3$) is termed \emph{triangle} and a graph is called \emph{triangle-free} if it contains no induced triangle. The length of a shortest cycle in $G$ is called the \emph{girth} of $G$. For a graph $G$,  $c(G)$ denotes the number of components in $G$. A set of vertices $S$ of a graph $G$ such that $c(G \backslash S) > c(G)$ is called a \emph{cut-set}.  A vertex $v$ is called a \emph{cut-vertex} if $\{v\}$ is a cut-set. A graph is \emph{2-connected} if its cut-sets have at least $2$ vertices.

A \emph{matching} in a graph $G$ is a set $M \subseteq E(G)$ of pairwise nonadjacent edges of $G$. A vertex $v$ of $G$ is \emph{saturated by} $M$ if $v \in V(M)$ and \emph{exposed by} $M$ otherwise. A matching $M$ is called \emph{maximal} in $G$ if there is no other matching of $G$ that contains $M$. A matching is called a \emph{maximum} matching of $G$ if it is a matching of maximum size. A matching $M$ in $G$ is a \emph{perfect matching} if $M$ saturates all vertices in $G$, i.e., $V(M) = V(G)$. For a vertex $v$, a matching $M$ is called a \emph{matching isolating} $v$ if $\{v\}$ is a component of $G \backslash V(M)$. A matching $M$ isolating a vertex $v$ is called \emph{minimal} if no subset of $M$ isolates $v$. A graph $G$ is \emph{equimatchable} if every maximal matching of $G$ is a maximum matching, i.e., every maximal matching has the same cardinality. A graph $G$ is \emph{randomly matchable} if it is an equimatchable graph admitting a perfect matching. A graph $G$ is \emph{factor-critical} if $G - v$ has a perfect matching for every vertex $v$ of G. For the sake of brevity, we denote equimatchable factor-critical graphs by \emph{EFC-graphs}.

The characterization of randomly matchable graphs was provided  by Sumner as follows:

\begin{theorem}\label{sumner}\textbf{(\cite{Sumner})}
A connected graph is randomly matchable if and only if it is isomorphic to $K_{2n}$ or $K_{n,n}$, $n \geq 1$.
\end{theorem}


A characterization of  EFC-graphs having a cut-set of size one was provided by Favaron as follows:

\begin{theorem}\label{fa-thm1}\textbf{(\cite{Favaron86})}
  A graph with vertex-connectivity $1$ is an EFC-graph if and only if:
  \begin{enumerate}
    \item  There exists exactly one cut-vertex $v$,
    \item  Every component of $G-v$ is randomly matchable, i.e. either $K_{2n}$ or $K_{n,n}$,
    \item  $v$ is adjacent to at least two adjacent vertices of every component of $G-v$.
  \end{enumerate}
\end{theorem}

For a $2$-connected EFC-graph $G$, the structure of the graph $G \backslash (V (M) \cup \{v\})$,  where $v \in V(G)$ and $M$ is a minimal matching isolating $v$, was given in \cite{Eiben13} as follows:

\begin{theorem}\label{eiben-thm2}\textbf{(\cite{Eiben13})}
  Let $G$ be a $2$-connected EFC-graph. Let $v$ be a vertex of $G$ and $M$ be a minimal matching isolating $v$. Then $G \backslash (V (M) \cup \{v\})$ is isomorphic to $K_{2n}$ or $K_{n,n}$ for some nonnegative integer $n$.
\end{theorem}

The following results are used frequently in our arguments and are key to analyzing equimatchable graphs with no triangle.

\begin{lemma}\label{independent-set}
A factor-critical graph $G$ is equimatchable if and only if there is no independent set $I$ of size $3$ such that $G \backslash I$ has a perfect matching.
\end{lemma}
\begin{proof}
Note that a factor-critical graph $G$ has a maximum matching of size $[|V(G)|-1]/2$.

$(\Rightarrow)$
Let $G$ be a EFC-graph. Then all maximal matchings of $G$ have size $[|V(G)| - 1]/2$. Assume to the contrary that there is an independent set $I$ of $G$ with $3$ vertices such that $G \backslash I$ has a perfect matching $M$. It implies that $M$ is a maximal matching of $G$ of size $[|V(G)|-3]/2$, contradiction. Therefore, we conclude that for every independent set $I$ of size $3$, $G \backslash I$ has no perfect matching.

$(\Leftarrow)$
Suppose, conversely, that for all independent sets of size $3$, $G \backslash I$ has no perfect matching. Assume that $G$ is not equimatchable. That is, there exists a maximal matching $M$ of size strictly less than $[|V(G)|-1]/2$ in $G$. Since $M$ is not a maximum matching by factor-criticality of $G$, by repetitively using augmenting paths, a maximal matching $M^*$ of size $([|V(G)|-1]/2)-1=[|V(G)|-3]/2$ in $G$ can be obtained. This implies that $I = G \backslash V(M^*)$ is an independent set of size $3$ and $M^*$ is a perfect matching in $G \backslash I$, contradiction. We conclude that $G$ is equimatchable.

\hfill\(\qed\)
\end{proof}

\begin{lemma}\label{removing-matching}
Let $G$ be a connected EFC-graph and $M$ be a matching of $G$. Then $G \backslash V(M)$ is equimatchable and factor-critical, and the followings hold:
\begin{enumerate}
  \item $G \backslash V(M)$ has exactly one odd component and this component is equimatchable and factor-critical.
  \item Even components of $G \backslash V(M)$ are randomly matchable.
\end{enumerate}
\end{lemma}
\begin{proof}
Let $G$ be a connected EFC-graph and $M$ be a matching of $G$. That is, every maximal matching of $G$ leaves exactly one vertex exposed.
  \begin{enumerate}
    \item \label{item1} Notice that $G$ is odd by factor-criticality and $V(M)$ has an even number of vertices. It follows that $G \backslash V(M)$ has an odd number of vertices. Hence, it is easy to see that $G \backslash V(M)$ contains at least one odd component. Indeed, $G \backslash V(M)$ contains exactly one odd component. Otherwise, if $G \backslash V(M)$ has $k$ odd components where $k \geq 2$, then every maximal matching extending $M$ leaves at least $k$ exposed vertices, contradicting with the fact that $G$ is an EFC-graph. Let $G^*$ be the unique odd component of $G \backslash V(M)$. Assume to the contrary that there exists a maximal matching $M^*$ of $G^*$ leaving at least three exposed vertices. It follows that every maximal matching of $G$ extending $M \cup M^*$ leaves at least three exposed vertices, contradicting with the fact that $G$ is an EFC-graph. Hence, every maximal matching of $G^*$ leaves exactly one vertex exposed; that is, $G^*$ is equimatchable and factor-critical.
    \item \label{item2} For an even component $H$ of $G \backslash V(M)$, assume to the contrary that there exists a maximal matching $M_H$ of $H$ leaving at least two exposed vertices. It follows that any maximal matching of $G$ extending $M \cup M_H$ leaves at least two exposed vertices, contradicting with  the fact that $G$ is an EFC-graph. Hence, every maximal matching of $H$ leaves no vertex exposed, i.e., every maximal matching of $H$ is a perfect matching. Hence, $H$ is randomly matchable.
  \end{enumerate}
  By combining \ref{item1} and \ref{item2}, one can easily observe that $G \backslash V(M)$ is equimatchable and factor-critical. \hfill\(\qed\)
\end{proof}

\begin{lemma}\label{ThirdVertex}
Let $G$ be a triangle-free graph and $C$ be an induced odd cycle  in $G$. For any two vertices $x$ and $y$ in $G\backslash C$, there exists a vertex $v \in C$ which is adjacent to neither $x$ nor $y$.
\end{lemma}
\begin{proof}
Let $G$ be a triangle-free graph with an induced cycle $C$ of length $2n+1$. Let $x$ and $y$ be any two vertices in $G\backslash C$. By triangle-freeness of $G$, the vertex $x$ (and equivalently $y$) can be adjacent to at most $n$ vertices in $C$. Hence, the size of the neighborhood of $x$ and $y$ in $C$ can be at most $2n$, implying that there exists a vertex $v \in C$ which is adjacent to neither $x$ nor $y$.
\hfill\(\qed\)
\end{proof}

\section{Structure of Triangle-Free Equimatchable Graphs}
This section is devoted to describe the structure of triangle-free equimatchable graphs, i.e., equimatchable graphs with girth at least $4$. Since a graph is equimatchable if and only if each of its components is equimatchable, it suffices to characterize connected triangle-free equimatchable graphs. To characterize triangle-free equimatchable graphs, we only need to investigate equimatchable graphs with girth exactly $4$ since the following structural characterization of the equimatchable graphs with girth at least five is provided in \cite{Frendrup}:

\begin{theorem}\label{frendrup}\textbf{(\cite{Frendrup})}
Let $G$ be a connected equimatchable graph with girth at least $5$. Then $G \in \cal{F} \cup \{C_{5}, C_{7}\}$, where $\cal{F}$ is the family of graphs containing $K_2$ and all connected bipartite graphs with bipartite sets $V_1$ and $V_2$ such that all vertices in $V_1$ are stems and no vertex from $V_2$ is a stem.
\end{theorem}

The class of triangle-free equimatchable graphs can be separated into two complementary subclasses, namely non-factor-critical and factor-critical graphs. First, we consider non-factor-critical triangle-free equimatchable graphs. A structural result for equimatchable graphs, not necessarily triangle-free, was given as follows:

\begin{theorem}\label{odd-hole-free}\text{(Lemma 8 in \cite{Dibek})}
  If $G$ is an equimatchable graph with an induced subgraph $C$ isomorphic to a cycle $C_{2k+1}$ for some $k \geq 2$, then $G$ is factor-critical.
\end{theorem}

\noindent
In other words, if an equimatchable graph is non-factor-critical, then it cannot have an induced cycle $C_{2k+1}$ for $k \geq 2$. Hence, it is easy to observe that a non-factor-critical equimatchable graph with no triangle is bipartite. In \cite{Lesk}, a characterization for equimatchable bipartite graphs was given as in the following way:

\begin{theorem}\label{lesk-thm-bipartiteEqm}\text{(\cite{Lesk})}
A connected bipartite graph $G = (U \cup V, E)$ with $|U| \leq |V|$ is equimatchable if and only if for all $u \in U$, there exists a non-empty $X \subseteq N(u)$ such that $|N(X)| \leq |X|$.
\end{theorem}

Therefore, it suffices to concentrate only on triangle-free factor-critical equimatchable graphs, shortly triangle-free EFC-graphs, in the rest of this paper.  More precisely, we focus on triangle-free equimatchable graphs with girth exactly $4$ and containing an odd cycle. Note that a factor-critical graph can not be bipartite, since if we choose a vertex from the small partite set (or from any partite set if their cardinalities are equal) there can not be a perfect matching in the rest of the graph. In the next lemma, we show that a triangle-free EFC-graph cannot have a cut vertex. That is, all graphs we deal with in this paper are indeed $2$-connected.

\begin{lemma}\label{cut-vertex}
Let $G$ be a connected triangle-free EFC-graph. Then $G$ is $2$-connected, i.e. $G$ has no cut-vertex.
\end{lemma}

\begin{proof}
Let $G$ be a connected triangle-free EFC-graph. Assume that $G$ has a cut-vertex $v$. By Theorem \ref{fa-thm1}, $v$ must be adjacent to at least two adjacent vertices in each component of $G-v$, which contradicts with the assumption that $G$ is triangle-free. \hfill\(\qed\)
\end{proof}

By Theorem \ref{frendrup}, it is easy to see that connected triangle-free EFC-graphs correspond to disjoint union of $C_5$, $C_7$, and equimatchable graphs with girth exactly $4$ and containing an odd cycle. On the other hand, a family of forbidden  induced subgraphs of equimatchable graphs was provided in \cite{Dibek} as follows:

\begin{theorem}\label{C-odd-free}\text{(\cite{Dibek})}
  Equimatchable graphs are $C_{2k+1}$-free for any $k \geq 4$.
\end{theorem}

By the following result, it only remains to analyze the structure of connected triangle-free EFC-graphs with an induced odd cycle of length five or seven:

\begin{corollary}\label{c5c7}
Let $G$ be a connected triangle-free EFC-graph. Then $G$ contains an induced odd cycle $C$ of length five or seven. Particularly, $G$ does not contain an induced odd cycle of any other length.
\end{corollary}

Throughout the rest of the paper, $G$ denotes a connected triangle-free EFC-graph and $C$ denotes an induced odd cycle in $G$. By Corollary \ref{c5c7}, $C$ is either $C_5$ or $C_7$. Since both $G$ and $C$  are  odd, $G \backslash C$ has an even number of vertices. Note here that $G\backslash C$ is allowed to be empty, i.e., $G$ can be $C_5$ or $C_7$ itself. In the next theorem, we prove that the removal of an odd cycle from $G$ results in a randomly matchable triangle-free graph, i.e., disjoint union of $K_{n,n}$ for some nonnegative integers $n$.

\begin{theorem}\label{KnnProof}
Let $G$ be a connected triangle-free EFC-graph. For an induced odd cycle $C$ in $G$, every component of $G\backslash C$ is isomorphic to $K_{n,n}$ for some nonnegative integer $n$.
\end{theorem}

\begin{proof}
Let $G$ be a connected triangle-free EFC-graph. By Corollary \ref{c5c7}, there exists an induced odd cycle $C$ of length $5$ or $7$ in $G$. If  $G$ is $C_5$ or $C_7$, we are done. If not, let $M$ be a minimal matching isolating a vertex $v$ in $C$ such that $M$ contains exactly two edges of $C$; that is, the neighbors of $v$ in $C$ are saturated by the edges from $C$. By Lemma \ref{cut-vertex}, $G$ is $2$-connected.
Hence, by Theorem \ref{eiben-thm2} and by triangle-freeness of $G$, $G \backslash (V (M) \cup \{v\})$ is isomorphic to $K_{n,n}$ for some positive integer $n$. Notice that if $C$ is the cycle of length seven, then $G \backslash (V (M) \cup \{v\})$ contains an edge of $C$.
We define a matching $M' \subset M$  containing all edges of $M$ except the edges of $C$, i.e., $M'= M\backslash E(C)$. Therefore, in order to describe the structure of $G\backslash C$, it is sufficient to show how edges of $M'$ are connected to the edges of $G\backslash C$ .

\begin{claim}\label{xy-x'y'}
Any two edges of $M'$ induce either a $2K_2$ or a $C_4$ in $G$.
\end{claim}

\textbf{\emph{Proof of Claim \ref{xy-x'y'}}}
Let $xy$ and $x'y'$ be two edges of $M'$ such that $x$ and $x'$ are adjacent to a vertex $v$ of $C$. If there exists no edge joining an endpoint of $xy$ and an endpoint of $x'y'$ in $G$, then the edges $xy$ and $x'y'$ form a $2K_2$. Suppose that there exists an edge joining an endpoint of $xy$ and an endpoint of $x'y'$ in $G$. Note that the vertices $x$ and $x'$ cannot be adjacent by triangle-freeness of $G$. Let ${v}^*$ be a neighbor of $v$ in $C$. Note that the vertices ${v}^*$, $x$ and $x'$ form an independent set of size three. If $y$ and $y'$ are adjacent, then $G\backslash \{{v}^*, x, x'\}$ has a perfect matching which contains the edge $yy'$, edges of $M'\backslash\{xy, x'y'\}$, a perfect matching of $C\backslash {v}^*$ and a perfect matching of randomly matchable graph $G \backslash (V (M) \cup \{v\})$. This contradicts with Lemma \ref{independent-set}. Hence, we deduce that $y$ and $y'$ are nonadjacent in $G$.

Now, assume that there exists an edge $xy'$ but no edge $x'y$ in $G$. By Lemma \ref{ThirdVertex}, there exists a vertex $v'$ in $C$ which is adjacent to neither $x'$ nor $y$. Then the vertices $v'$, $x'$ and $y$ form an independent set of size three and $G\backslash \{v', x', y\}$ has a perfect matching containing the edge $xy'$, edges of $M'\backslash\{xy, x'y'\}$, a perfect matching of $C\backslash v'$ and a perfect matching of randomly matchable graph $G \backslash (V (M) \cup \{v\})$, contradicting with Lemma \ref{independent-set}. Hence, $x'y$ is also an edge so that the vertices $x$, $y$, $x'$, and $y'$ induce a $C_4$.
\hfill\(\Box\)

\begin{claim}\label{xy-Knn}
An edge of $M'$ and each edge of the subgraph $G \backslash V(M)$ induce either a $2K_2$ or a $C_4$ in $G\backslash C$.
\end{claim}

\textbf{\emph{Proof of Claim \ref{xy-Knn}}}
By Theorem \ref{eiben-thm2}, the subgraph $G \backslash (V (M) \cup \{v\})$  is isomorphic to $K_{n,n}$ for some positive integer $n$. Since $v \in C$, $G \backslash (V (M) \cup \{v\}) = G \backslash V(M)$ in $G\backslash C$.
For simplicity, let us denote $G \backslash (V (M) \cup \{v\})$ by $K_{n,n}$ and $G \backslash V(M)$ in $G\backslash C$ by $K_{n,n} \backslash C$. Note that $G \backslash (V (M) \cup \{v\})$ contains either no edge or one edge of $C$ if $C$ is the cycle of length five or seven, respectively.
If neither of the endpoints of an edge of $M'$ is adjacent to a vertex of $K_{n,n}$, then such an edge and an edge of $K_{n,n}$ induce a $2K_2$. Now let $xy$ be an edge of $M'$ such that there exists an edge in $G$ joining an endpoint of $xy$ to a vertex of $K_{n,n}\backslash C$. Without loss of generality, suppose that the vertices $y$ and $w$ are adjacent in $G$, where $w$  is a vertex in $K_{n,n} \backslash C$.
Then we need to show that the vertex $x$ is adjacent to all neighbors of $w$ in $K_{n,n}\backslash C$. Assume to the contrary that there exists  a neighbor $z$ of $w$ in $K_{n,n}\backslash C$ such that $x$ and $z$ are nonadjacent. By Lemma \ref{ThirdVertex}, there exists a vertex $v' \in C$ which is adjacent to neither $x$ nor $z$. Then the vertices $v'$, $x$ and $z$ form an independent set of size three and $G\backslash \{v', x, z\}$ has a perfect matching containing the edge $wy$, edges of $M'\backslash\{xy\}$, a perfect matching of $C\backslash v'$  and a perfect matching of randomly matchable graph $(K_{n,n}\backslash C)\backslash {wz}$, contradicting with Lemma \ref{independent-set}. Hence, we deduce that there exists an edge joining $x$ and $z$ in $G$. This implies that $x$ is adjacent to all neighbors of $w$ in $K_{n,n}\backslash C$. By symmetry, one can observe that the vertex $y$ is also adjacent to all neighbors of $z$ in $K_{n,n}\backslash C$. By triangle-freeness of $G$, we conclude that  $K_{n,n}\backslash C$ and the edge $xy$ form a randomly matchable graph $K_{n+1,n+1}\backslash C$ in $G\backslash C$; that is, $xy$ and each edge of $K_{n,n}\backslash C$ induce a $C_4$ in $G$.
\hfill\(\Box\)

\begin{claim}\label{xy-x'y'-wz}
If two edges of $M'$ induce a $C_4$ with the same edge in $G\backslash C$, then these two edges of $M'$ induce also a $C_4$ in $G$.
\end{claim}

\textbf{\emph{Proof of Claim \ref{xy-x'y'-wz}}}
Let $xy$ and $x'y'$ be two edges of $M'$, $K_{n,n}$ represent the subgraph $G \backslash (V (M) \cup \{v\})$ and let $wz$ be an edge in $K_{n,n}\backslash C$. Suppose that both $xy$ and $x'y'$  induce a $C_4$ with the edge $wz$ in $K_{n,n}$. Without loss of generality, assume that the vertices $x$ and $x'$ are adjacent to the vertex $w$  and the vertices $y$ and $y'$ are adjacent to the vertex $z$ in $G$. Note here that $x$ cannot be adjacent to $x'$, and $y$ cannot be adjacent to $y'$ by triangle-freeness of $G$. Assume $xy' \notin E(G)$. By Lemma \ref{ThirdVertex}, there exists a vertex $v' \in C$ which is adjacent to neither $x$ nor $y'$. Then the vertices $v'$, $x$ and $y'$ form an independent set of size three and $G\backslash \{v', x, y'\}$ has a perfect matching containing the edges $x'w$, $zy$, edges of $M'\backslash\{xy, x'y'\}$, a perfect matching of $C\backslash v'$  and a perfect matching of $(K_{n,n}\backslash)\backslash wz$. Since it contradicts with Lemma \ref{independent-set}, we deduce that there exists an edge joining $x$ and $y'$ in $G$. By symmetry, $xy' \in E(G)$.
\hfill\(\Box\)
\\
\noindent

By combining Claim \ref{xy-x'y'}, \ref{xy-Knn} and \ref{xy-x'y'-wz}, it is easy to verify that  an edge of $M'$ induces either a $2K_2$ or a $C_4$ with all edges of a component of $G \backslash C$. Hence, we conclude that each component of $G \backslash C$ is isomorphic to a $K_{n,n}$ for some positive integer $n$.
\hfill\(\qed\)
\end{proof}

\begin{corollary}\label{only-one-odd-cycle}
Let $G$ be a connected triangle-free EFC-graph. Then, there are no two vertex-disjoint induced odd cycles in G
and the length of any induced odd cycle is either five or seven.
\end{corollary}

In order to show how components of $G\backslash C$ are connected to $C$, we first focus on how each edge of $G\backslash C$ is connected to $C$. The next lemma guarantees that each edge of $G\backslash C$ induces a $C_4$ with at least one edge of $C$ in $G$.

\begin{lemma}\label{EdgeEdge}
Let $G$ be a connected triangle-free EFC-graph and $C$ be an induced odd cycle in $G$. Then each edge of $G\backslash C$ induces a $C_4$ with at least one edge of $C$. Furthermore, the neighborhood of the endpoints of an edge of $G\backslash C$ in $C$ induce a path.
\end{lemma}
\begin{proof}
Let $G$ be a connected  triangle-free EFC-graph, $C$ be an induced odd cycle in $G$, and $H$ be a component of $G\backslash C$. By Theorem \ref{KnnProof}, $H$  is isomorphic to $K_{n,n}$  for some nonnegative integer $n$. Consider an edge joining a vertex of $H$ and a vertex of $C$, say the edge $xv$, where $x \in H$ and $v \in C$. Such an edge always exists in $G$ since $G$ is connected. Then we need to show that all neighbors of $x$ in $H$ are adjacent to a neighbor of $v$ in $C$. Assume to the contrary that there exists a neighbor of $x$ in $H$, say $y$, that is adjacent to none of the neighbors $v'$ and $v''$ of $v$ in $C$.
Then the vertices $y$, $v'$ and $v''$ form an independent set of size three and $G\backslash \{y, {v}', v''\}$ has a perfect matching containing the edges $xv$, a perfect matching of the path $C\backslash\{v,{v}', v''\}$ which has an even number of vertices, and a perfect matching of randomly matchable graph $G\backslash (C \cup \{x,y\})$. Since it contradicts with  Lemma \ref{independent-set}, we deduce that there exists an edge joining $y$ and at least one of $\{v', v''\}$. Hence, all neighbors of $x$ in $H$ are adjacent to at least one of the vertices $v'$ and $v''$. Without loss of generality, assume that $y$ is adjacent to $v'$ in $C$. In a similar way, one can observe that all neighbors of $y$ in $H$ are adjacent to at least one neighbor of $v'$ in $C$. Therefore, we conclude that each edge in $H$ induce a $C_4$ with at least one edge in $C$ since $H$ is isomorphic to $K_{n,n}$  for some nonnegative integer $n$. In addition, it follows from the above discussion that the neighborhood of the endpoints of an edge of $G\backslash C$ in $C$ induce a path, as desired.
\hfill\(\qed\)
\end{proof}

The next lemma, which describes induced triangle-free EFC subgraphs of $G$, is the primary tool for our characterization of triangle-free EFC-graphs.

\begin{lemma}\label{InducedEFCSubgraph}
Let $G$ be a connected triangle-free EFC-graph and $C$ be an induced odd cycle in $G$. Then every subgraph of $G$ induced by the vertices of $C$ together with a randomly matchable subgraph of $G\backslash C$ is a connected triangle-free EFC-graph.
\end{lemma}
\begin{proof}
Let $G$ be a connected triangle-free EFC-graph, $C$ be an induced odd cycle in $G$, and $H$ be a randomly matchable subgraph of $G\backslash C$. By Theorem \ref{KnnProof}, every component of $G\backslash C$ is isomorphic to $K_{n,n}$ for some nonnegative integer $n$. Then $(G\backslash C)\backslash H$ is also a randomly matchable subgraph of $G\backslash C$, say $H'$. Hence, a subgraph of $G$ induced by the vertices of $C$ and $H$ can be considered as an induced subgraph $G\backslash V(M)$ of $G$, where $M$ is a perfect matching of $H'$. Let $G'$ be such an induced subgraph. By Lemma \ref{EdgeEdge}, $G'$ is connected. By Lemma \ref{removing-matching}, $G'$ is a triangle-free EFC graph. This completes the proof.
\hfill\(\qed\)
\end{proof}
\noindent
Therefore, each  subgraph of $G$ induced by the vertices of $C$ and an edge of $G\backslash C$ is indeed a connected triangle-free EFC-graph.

\begin{corollary}\label{C5orC7andK2}
The graphs $C_5 \cup K_2$ and  $C_7 \cup K_2$ are forbidden subgraphs for connected triangle-free EFC graphs.
\end{corollary}

The next theorem provides the structural characterization of connected triangle-free EFC-graphs containing only one edge after the removal of an induced odd cycle. Note here that such graphs are the smallest connected triangle-free EFC-graphs except $C_5$ and $C_7$.

\begin{theorem}\label{C-OneEdge-Structure}
Let $G$ be a connected  triangle-free EFC-graph and $C$ be an induced odd cycle in $G$. Then a subgraph of $G$ induced by the vertices of $C$ and an edge in $G\backslash C$ is isomorphic to one of the graphs given in Figure \ref{fig:InducedOneEdgeSubgraph}. \noindent
\begin{figure}[hbtp]
\centering
\begin{tikzpicture}[scale=0.7]
\filldraw [black] (-0.7,0) circle (3pt);
\filldraw [black] (0.7,0) circle (3pt);
\filldraw [black] (-0.7,1.3) circle (3pt);
\filldraw [black] (0.7,1.3) circle (3pt);
\filldraw [black] (-1.2,-1) circle (3pt);
\filldraw [black] (1.2,-1) circle (3pt);
\filldraw [black] (0,-2) circle (3pt);
\filldraw [white] (0,-2.5) circle (3pt);

\node at (0,-3.5) {\textbf{(a)}};

\draw[very thick] (-0.7,0) -- (0.7,1.3) -- (-0.7,1.3) -- (0.7,0) -- (-0.7,0) -- (-1.2,-1) -- (0,-2) -- (1.2,-1) -- (0.7,0);
\end{tikzpicture}
\hskip 0.8cm
\begin{tikzpicture}[scale=0.7]
\filldraw [black] (-0.7,0) circle (3pt);
\filldraw [black] (0.7,0) circle (3pt);
\filldraw [black] (-0.7,1.3) circle (3pt);
\filldraw [black] (0.7,1.3) circle (3pt);
\filldraw [black] (-1.2,-1) circle (3pt);
\filldraw [black] (1.2,-1) circle (3pt);
\filldraw [black] (0,-2) circle (3pt);
\filldraw [white] (0,-2.5) circle (3pt);

\node at (0,-3.5) {\textbf{(b)}};

\draw[very thick] (-0.7,0) -- (0.7,1.3) -- (-0.7,1.3) -- (0.7,0) -- (-0.7,0) -- (-1.2,-1) -- (0,-2) -- (1.2,-1) -- (0.7,0);
\draw[very thick] (-1.2,-1) .. controls (-1,0.5) .. (-0.7,1.3);
\end{tikzpicture}
\hskip 0.8cm
\begin{tikzpicture}[scale=0.7]
\filldraw [black] (-0.7,0) circle (3pt);
\filldraw [black] (0.7,0) circle (3pt);
\filldraw [black] (-0.7,1.3) circle (3pt);
\filldraw [black] (0.7,1.3) circle (3pt);
\filldraw [black] (-1.2,-1) circle (3pt);
\filldraw [black] (1.2,-1) circle (3pt);
\filldraw [black] (0,-2) circle (3pt);
\filldraw [white] (0,-2.5) circle (3pt);

\node at (0,-3.5) {\textbf{(c)}};

\draw[very thick] (-0.7,0) -- (0.7,1.3) -- (-0.7,1.3) -- (0.7,0) -- (-0.7,0) -- (-1.2,-1) -- (0,-2) -- (1.2,-1) -- (0.7,0);
\draw[very thick] (-1.2,-1) .. controls (-1,0.5) .. (-0.7,1.3);
\draw[very thick] (1.2,-1) .. controls (1,0.5) .. (0.7,1.3);
\end{tikzpicture}
\hskip 0.8cm
\begin{tikzpicture}[scale=0.7]
\filldraw [black] (-0.7,0) circle (3pt);
\filldraw [black] (0.7,0) circle (3pt);
\filldraw [black] (-0.7,1.3) circle (3pt);
\filldraw [black] (0.7,1.3) circle (3pt);
\filldraw [black] (-1.2,-1) circle (3pt);
\filldraw [black] (1.2,-1) circle (3pt);
\filldraw [black] (0.9,-2) circle (3pt);
\filldraw [black] (-0.9,-2) circle (3pt);
\filldraw [black] (0,-2.5) circle (3pt);

\node at (0,-3.5) {\textbf{(d)}};

\draw[very thick] (-0.7,0) -- (0.7,1.3) -- (-0.7,1.3) -- (0.7,0) -- (-0.7,0) -- (-1.2,-1) -- (-0.9,-2) -- (0,-2.5) -- (0.9,-2) -- (1.2,-1) -- (0.7,0);
\draw[very thick] (-1.2,-1) .. controls (-1,0.5) .. (-0.7,1.3);
\draw[very thick] (1.2,-1) .. controls (1,0.5) .. (0.7,1.3);
\end{tikzpicture}
\hskip 0.8cm
\begin{tikzpicture}[scale=0.7]
\filldraw [black] (-0.7,0) circle (3pt);
\filldraw [black] (0.7,0) circle (3pt);
\filldraw [black] (-0.7,1.3) circle (3pt);
\filldraw [black] (0.7,1.3) circle (3pt);
\filldraw [black] (-1.2,-1) circle (3pt);
\filldraw [black] (1.2,-1) circle (3pt);
\filldraw [black] (0.9,-2) circle (3pt);
\filldraw [black] (-0.9,-2) circle (3pt);
\filldraw [black] (0,-2.5) circle (3pt);

\node at (0,-3.5) {\textbf{(e)}};

\draw[very thick] (-0.7,0) -- (0.7,1.3) -- (-0.7,1.3) -- (0.7,0) -- (-0.7,0) -- (-1.2,-1) -- (-0.9,-2) -- (0,-2.5) -- (0.9,-2) -- (1.2,-1) -- (0.7,0);
\draw[very thick] (-1.2,-1) .. controls (-1,0.5) .. (-0.7,1.3);
\draw[very thick] (1.2,-1) .. controls (1,0.5) .. (0.7,1.3);
\draw[very thick] (-0.7,1.3) .. controls (2,3) and (2.3,-2) .. (0.9,-2);
\end{tikzpicture}
\caption{Subgraphs of $G$ induced by the vertices of $C$ and an edge of $G\backslash C$.} \label{fig:InducedOneEdgeSubgraph}
\end{figure}
\end{theorem}
\begin{proof}
Let $G$ be a connected  triangle-free EFC-graph and $C$ be an induced odd cycle in $G$. Recall that $C$ is a cycle of length five or seven. Let $G'$ be a subgraph of $G$ induced by the vertices of $C$ and an edge of $G\backslash C$. By Lemma \ref{InducedEFCSubgraph}, $G'$ is a connected  triangle-free EFC-graph with induced odd cycle $C$ such that $G'\backslash C$ is isomorphic to $K_2$. By Lemma \ref{EdgeEdge} and Lemma \ref{ThirdVertex}, it is easy to see that the number of vertices in $C$ that are adjacent to the vertices of $G'\backslash C$ can be at least two and at most four when $C$ is a cycle of length five, and at least two and at most six when $C$ is a cycle of length seven. Besides, all such vertices of $C$ are consecutive in either case of $C$. For a cycle $C$ of length five, we obtain the graphs $(a)$, $(b)$ and $(c)$ given in Figure \ref{fig:InducedOneEdgeSubgraph}. One can easily verify that all of these graphs are connected triangle-free EFC-graphs. On the other hand, for a cycle $C$ of length seven, one can observe that the only triangle-free EFC-graphs are the graphs $d$ and $e$ given in Figure \ref{fig:InducedOneEdgeSubgraph}, i.e., the graphs having four or five consecutive vertices in $C$ that are adjacent to the vertices of $G'\backslash C$. In fact, in all other graphs one can find an independent set of size three containing at least two vertices of $C$ such that the remaining graph after the removal of such an independent set has a perfect matching, contradicting with Lemma \ref{independent-set} and hence implying that such graphs are not EFC-graphs. This completes the proof.
\hfill\(\qed\)
\end{proof}

\begin{corollary}\label{AtLeastTwoFour}
Let $G$ be a connected triangle-free EFC-graph and $C$ be an induced odd cycle in $G$. The number of vertices in $C$ that are adjacent to the vertices of an edge in $G\backslash C$ is at least two or four, where $C$ is a cycle of length five or seven, respectively. Furthermore, all such vertices of $C$ are consecutive in either case of $C$.
\end{corollary}

In the next lemma, we show that if any two edges of $G\backslash C$ induce a $C_4$ with the same edge in $C$, then these two edges of $G\backslash C$ belong to the same component of $G\backslash C$.

\begin{lemma} \label{DisjointEdges1}
Let $G$ be a connected triangle-free EFC-graph and $C$ be an induced odd cycle in $G$. Then two edges from different components of $G\backslash C$ cannot induce a $C_4$ with the same edge in $C$.
\end{lemma}
\begin{proof}
Let $G$ be a connected triangle-free EFC-graph and $C$ be an induced odd cycle in $G$. Let $xy$ and $wz$ be two edges from different components of $G\backslash C$, i.e., $xy$ and $wz$ induce a $2K_2$ in $G\backslash C$. Let the consecutive vertices of $C$ be labeled by $v_1$, $v_2$, $v_3$, $v_4$, $v_5$, if necessary $v_6$ and $v_7$, depending on the length of $C$. Assume to the contrary that both edges $xy$ and $wz$ induce a $C_4$ with the edge $v_2v_3$, where $x$ and $w$ are adjacent to $v_2$, and $y$ and $z$ are adjacent to $v_3$. If $x$ and $v_4$ are nonadjacent in $G$, the vertices $x$, $z$ and $v_4$ form an independent set of size $3$. Then $G\backslash \{x, z, v_4\}$ has a perfect matching containing the edges $wv_2$, $yv_3$, a perfect matching of the path $C\backslash \{v_2,v_3,v_4\}$ which has an even number of vertices, and a perfect matching of randomly matchable graph $G\backslash (C \cup \{x,y,w,z\})$. Since it contradicts with Lemma \ref{independent-set}, we deduce that $x$ and $v_4$ are adjacent in $G$. By symmetry, we also deduce that $z$ and $v_1$ are adjacent in $G$.

Suppose that $C$ is a cycle of length five in $G$. If $v_5$ is adjacent to neither $w$ nor $y$ in $G$, the vertices $w$, $y$ and $v_5$ form an independent set of size $3$. Then  $G\backslash \{w, y, v_5\}$ has a perfect matching containing the edges $xv_4$, $zv_1$, $v_2v_3$ and a perfect matching of randomly matchable graph $G\backslash (C \cup \{x,y,w,z\})$.
By  Lemma \ref{independent-set}, it follows  that $v_5$ is adjacent to at least one of $\{w,y\}$. W.l.o.g., we assume that $v_5$ is adjacent to $w$. Then, the vertices $w$, $y$ and $v_4$ form an independent set of size $3$ and $G\backslash \{w, y, v_4\}$ has a perfect matching containing the edges $xv_2$, $zv_3$, $v_1v_5$, and a perfect matching of randomly matchable graph $G\backslash (C \cup \{x,y,w,z\})$, contradicting  Lemma \ref{independent-set}. Hence, the edges $xy$ and $wz$ cannot induce a $C_4$ with the same edge of $C$, where $C$ is a cycle of length five.

Suppose that $C$ is a cycle of length seven. If $v_7$ is adjacent to neither $w$ nor $y$ in $G$, the vertices $w$, $y$ and $v_7$ form an independent set of size $3$. Then $G\backslash \{w, y, v_7\}$ has a perfect matching containing the edges $zv_1$, $xv_4$, $v_2v_3$, $v_5v_6$, and a perfect matching of randomly matchable graph $G\backslash (C \cup \{x,y,w,z\})$.
That is, by Lemma \ref{independent-set}, $v_7$ is adjacent to at least one of  $\{w,y\}$. In the case where the edge $yv_7$ exists, the vertices $w$, $y$ and $v_1$ form an independent set of size $3$ and  $G\backslash \{w, y, v_1\}$ has a perfect matching containing the edges $xv_2$, $zv_3$, $v_4v_5$, $v_6v_7$, and a perfect matching of randomly matchable graph $G\backslash (C \cup \{x,y,w,z\})$.  Then, Lemma \ref{independent-set} implies that the vertex $v_7$ is adjacent to $w$ but not to $y$ in $G$. By symmetry, we also deduce that  the vertex $v_5$ is adjacent to $y$ but not to $w$ in $G$. It follows that the vertices $w$, $y$ and $v_6$ form an independent set of size $3$ and  $G\backslash \{w, y, v_6\}$ has a perfect matching containing the edges $zv_3$, $xv_2$, $v_1v_7$, $v_4v_5$, and a perfect matching of randomly matchable graph $G\backslash (C \cup \{x,y,w,z\})$, a contradiction with Lemma \ref{independent-set}. Therefore, we conclude that both $xy$ and $wz$ cannot induce a $C_4$ with the same edge of $C$, where $C$ is a cycle of length seven.
\hfill\(\qed\)
\end{proof}

By Lemma \ref{InducedEFCSubgraph}, a subgraph of $G$ induced by the vertices of $C$ and  two edges from different components of $G\backslash C$ is also a connected triangle-free EFC-graph. From this viewpoint, the next theorem provides the structural characterization of connected triangle-free EFC-graphs containing only a $2K_2$ after the removal of an induced odd cycle.

\begin{theorem} \label{C-TwoEdge-Structure}
Let $G$ be a connected triangle-free EFC-graph and $C$ be an induced odd cycle in $G$. Then a subgraph of $G$ induced by the vertices of $C$ and two edges from different components of $G\backslash C$ is isomorphic to one of the graphs given in Figure \ref{fig:InducedTwoEdgesSubgraph}.
\begin{figure}[hbtp]
\begin{tikzpicture}[scale=0.6]
\filldraw [black] (-0.7,0) circle (3pt);
\filldraw [black] (0.7,0) circle (3pt);
\filldraw [black] (-0.7,1) circle (3pt);
\filldraw [black] (0.7,1) circle (3pt);
\filldraw [black] (-1.4,-1) circle (3pt);
\filldraw [black] (1.4,-1) circle (3pt);
\filldraw [black] (1,-2) circle (3pt);
\filldraw [black] (-1,-2) circle (3pt);
\filldraw [black] (0,-2.5) circle (3pt);
\filldraw [black] (-1.7,-2.7) circle (3pt);
\filldraw [black] (-0.8,-3.2) circle (3pt);
\filldraw [white] (-0.8,-3.6) circle (3pt);

\node at (0,-4) {\textbf{(A)}};

\draw[very thick] (-0.7,0) -- (0.7,1) -- (-0.7,1) -- (0.7,0) -- (-0.7,0) -- (-1.4,-1) -- (-1,-2) -- (0,-2.5) -- (1,-2) -- (1.4,-1) -- (0.7,0);
\draw[very thick] (-1.4,-1) .. controls (-1,0.5) .. (-0.7,1);
\draw[very thick] (1.4,-1) .. controls (1,0.5) .. (0.7,1);
\draw[very thick] (-1.7,-2.7) -- (-0.8,-3.2);
\draw[very thick] (-1.7,-2.7) -- (0,-2.5);
\draw[very thick] (-1.7,-2.7) .. controls (-1.7,-2) .. (-1.4,-1);
\draw[very thick] (-0.8,-3.2) -- (-1,-2);
\draw[very thick] (-0.8,-3.2) .. controls (0.6,-2.6) .. (1,-2);
\end{tikzpicture}
\hskip 0.4cm
\begin{tikzpicture}[scale=0.6]
\filldraw [black] (-0.7,0) circle (3pt);
\filldraw [black] (0.7,0) circle (3pt);;
\filldraw [black] (-0.7,1) circle (3pt);
\filldraw [black] (0.7,1) circle (3pt);
\filldraw [black] (-1.4,-1) circle (3pt);
\filldraw [black] (1.4,-1) circle (3pt);
\filldraw [black] (1,-2) circle (3pt);
\filldraw [black] (-1,-2) circle (3pt);
\filldraw [black] (0,-2.5) circle (3pt);
\filldraw [black] (-1.7,-2.7) circle (3pt);
\filldraw [black] (-0.8,-3.2) circle (3pt);
\filldraw [white] (-0.8,-3.6) circle (3pt);

\node at (0,-4) {\textbf{(B)}};

\draw[very thick] (-0.7,0) -- (0.7,1) -- (-0.7,1) -- (0.7,0) -- (-0.7,0) -- (-1.4,-1) -- (-1,-2) -- (0,-2.5) -- (1,-2) -- (1.4,-1) -- (0.7,0);
\draw[very thick] (-1.4,-1) .. controls (-1,0.5) .. (-0.7,1);
\draw[very thick] (1.4,-1) .. controls (1,0.5) .. (0.7,1);
\draw[very thick] (-1.7,-2.7) -- (-0.8,-3.2);
\draw[very thick] (-1.7,-2.7) -- (0,-2.5);
\draw[very thick] (-1.7,-2.7) .. controls (-1.7,-2) .. (-1.4,-1);
\draw[very thick] (-0.8,-3.2) -- (-1,-2);
\draw[very thick] (-0.8,-3.2) .. controls (0.6,-2.6) .. (1,-2);
\draw[very thick] (-0.7,1) .. controls (2.2,2.2) and (2.2,-1) .. (1,-2);
\end{tikzpicture}
\hskip 0.2cm
\begin{tikzpicture}[scale=0.6]
\filldraw [black] (-0.7,0) circle (3pt);
\filldraw [black] (0.7,0) circle (3pt);
\filldraw [black] (-0.7,1.3) circle (3pt);
\filldraw [black] (0.7,1.3) circle (3pt);
\filldraw [black] (-1.2,-1) circle (3pt);
\filldraw [black] (1.2,-1) circle (3pt);
\filldraw [black] (0,-2) circle (3pt);
\filldraw [black] (-2,-1.9) circle (3pt);
\filldraw [black] (-0.8,-2.9) circle (3pt);

\node at (0,-4) {\textbf{(C)}};

\draw[very thick] (-0.7,0) -- (0.7,1.3) -- (-0.7,1.3) -- (0.7,0) -- (-0.7,0) -- (-1.2,-1) -- (0,-2) -- (1.2,-1) -- (0.7,0);
\draw[very thick] (-1.2,-1) .. controls (-1,0.5) .. (-0.7,1.3);
\draw[very thick] (1.2,-1) .. controls (1,0.5) .. (0.7,1.3);
\draw[very thick] (-2,-1.9) -- (-0.8,-2.9);
\draw[very thick] (-2,-1.9) -- (0,-2);
\draw[very thick] (-0.8,-2.9) -- (-1.2,-1);
\end{tikzpicture}
\hskip 0.2cm
\begin{tikzpicture}[scale=0.6]
\filldraw [black] (-0.7,0) circle (3pt);
\filldraw [black] (0.7,0) circle (3pt);
\filldraw [black] (-0.7,1.3) circle (3pt);
\filldraw [black] (0.7,1.3) circle (3pt);
\filldraw [black] (-1.2,-1) circle (3pt);
\filldraw [black] (1.2,-1) circle (3pt);
\filldraw [black] (0,-2) circle (3pt);
\filldraw [black] (-2,-1.9) circle (3pt);
\filldraw [black] (-0.8,-2.9) circle (3pt);

\node at (0,-4) {\textbf{(D)}};

\draw[very thick] (-0.7,0) -- (0.7,1.3) -- (-0.7,1.3) -- (0.7,0) -- (-0.7,0) -- (-1.2,-1) -- (0,-2) -- (1.2,-1) -- (0.7,0);
\draw[very thick] (-1.2,-1) .. controls (-1,0.5) .. (-0.7,1.3);
\draw[very thick] (1.2,-1) .. controls (1,0.5) .. (0.7,1.3);
\draw[very thick] (-2,-1.9) -- (-0.8,-2.9);
\draw[very thick] (-2,-1.9) -- (0,-2);
\draw[very thick] (-0.8,-2.9) -- (-1.2,-1);
\draw[very thick] (-0.8,-2.9) .. controls (0.4,-2.4) .. (1.2,-1);
\end{tikzpicture}
\hskip 0.2cm
\begin{tikzpicture}[scale=0.6]
\filldraw [black] (-0.7,0) circle (3pt);
\filldraw [black] (0.7,0) circle (3pt);
\filldraw [black] (-0.7,1.3) circle (3pt);
\filldraw [black] (0.7,1.3) circle (3pt);
\filldraw [black] (-1.2,-1) circle (3pt);
\filldraw [black] (1.2,-1) circle (3pt);
\filldraw [black] (0,-2) circle (3pt);
\filldraw [black] (-2,-1.9) circle (3pt);
\filldraw [black] (-0.8,-2.9) circle (3pt);

\node at (0,-4) {\textbf{(E)}};

\draw[very thick] (-0.7,0) -- (0.7,1.3) -- (-0.7,1.3) -- (0.7,0) -- (-0.7,0) -- (-1.2,-1) -- (0,-2) -- (1.2,-1) -- (0.7,0);
\draw[very thick] (-1.2,-1) .. controls (-1,0.5) .. (-0.7,1.3);
\draw[very thick] (-2,-1.9) -- (-0.8,-2.9);
\draw[very thick] (-2,-1.9) -- (0,-2);
\draw[very thick] (-0.8,-2.9) -- (-1.2,-1);
\draw[very thick] (-0.8,-2.9) .. controls (0.4,-2.4) .. (1.2,-1);
\end{tikzpicture}
\hskip 0.2cm
\begin{tikzpicture}[scale=0.6]
\filldraw [black] (-0.7,0) circle (3pt);
\filldraw [black] (0.7,0) circle (3pt);
\filldraw [black] (-0.7,1.3) circle (3pt);
\filldraw [black] (0.7,1.3) circle (3pt);
\filldraw [black] (-1.2,-1) circle (3pt);
\filldraw [black] (1.2,-1) circle (3pt);
\filldraw [black] (0,-2) circle (3pt);
\filldraw [black] (-2,-1.9) circle (3pt);
\filldraw [black] (-0.8,-2.9) circle (3pt);

\node at (0,-4) {\textbf{(F)}};

\draw[very thick] (-0.7,0) -- (0.7,1.3) -- (-0.7,1.3) -- (0.7,0) -- (-0.7,0) -- (-1.2,-1) -- (0,-2) -- (1.2,-1) -- (0.7,0);
\draw[very thick] (1.2,-1) .. controls (1,0.5) .. (0.7,1.3);
\draw[very thick] (-2,-1.9) -- (-0.8,-2.9);
\draw[very thick] (-2,-1.9) -- (0,-2);
\draw[very thick] (-0.8,-2.9) -- (-1.2,-1);
\end{tikzpicture}
\caption{Subgraphs of $G$ induced by the vertices of $C$ and two edges from different components of $G\backslash C$.} \label{fig:InducedTwoEdgesSubgraph}
\end{figure}
\end{theorem}
\begin{proof}

Let $G$ be a connected  triangle-free EFC-graph and $C$ be an induced odd cycle in $G$. Recall that $C$ is a cycle of length five or seven. Let $G'$ be a subgraph of $G$ induced by the vertices of $C$ and two edges, say $xy$ and $wz$, from different components of $G\backslash C$. By Lemma \ref{InducedEFCSubgraph}, $G'$ is a connected  triangle-free EFC-graph with induced odd cycle $C$ such that $G'\backslash C$ is isomorphic to $2K_2$. Let $S_{xy}$ and $S_{wz}$ denote the set of vertices of $C$ that are adjacent to the vertices of $xy$ and $wz$, respectively. Note here that, by Corollary  \ref{AtLeastTwoFour}, the vertices in each of $S_{xy}$ and $S_{wz}$ are consecutive in $C$. Remark that, by Theorem \ref{C-OneEdge-Structure}, both subgraphs of $G'$ induced by the vertices of $C$ and each of $xy$ and $wz$ are also connected triangle-free EFC-graphs given in Figure \ref{fig:InducedOneEdgeSubgraph}.  That is, we only need to figure out how two independent edges $xy$ and $wz$ can form a graph of Figure \ref{fig:InducedOneEdgeSubgraph} with the same induced odd cycle. Furthermore, by Lemma \ref{DisjointEdges1}, an edge of $C$ cannot induce a $C_4$ with both edges $xy$ and $wz$, i.e., the vertices of $xy$ and $wz$ have no common consecutive neighbors in $C$. We then need to examine two disjoint and complementary cases: the vertices of $xy$ and $wz$ have common nonconsecutive neighbors in $C$ and have no common neighbor in $C$.

We first consider the case where the vertices of $xy$ and $wz$ have at least one common neighbor in $C$. For a cycle $C$ of length seven, by Theorem \ref{C-OneEdge-Structure}, each of $S_{xy}$ and $S_{wz}$  contains four or five vertices, see the graphs $(d)$ and $(e)$ in Figure \ref{fig:InducedOneEdgeSubgraph}. By  Lemma \ref{DisjointEdges1} and the cardinality of $C$, one can observe that at least one of the sets $S_{xy}$ and $S_{wz}$  have size four whereas the other may have size four or five. Hence, we deduce that $G'$ is isomorphic to one of the graphs $(A)$ and $(B)$ given in Figure \ref{fig:InducedTwoEdgesSubgraph}, where $C$ is a cycle of length seven. For a cycle $C$ of length five, by Theorem \ref{C-OneEdge-Structure}, each of $S_{xy}$ and $S_{wz}$  contains at least two and at most four vertices, see the graphs $(a)$, $(b)$ and $(c)$ in Figure \ref{fig:InducedOneEdgeSubgraph}. Then, it is easy to verify that the only possible ordered pairs $(|S_{xy}|,|S_{wz}|)$ are $(2,4)$, $(3,3)$ and $(3,4)$. Otherwise, in the other pairs either there exists an independent set $I$ of size three such that $G'\backslash I$ has a perfect matching, contradicting with Lemma \ref{independent-set}, or there exists an edge of $C$ forming a $C_4$ with two different edges in $G'\backslash C$, contradicting Lemma \ref{DisjointEdges1}. (Note that this later corresponds to the cardinalities (4,4) for respectively $S_{xy}$ and $S_{wz}$.) Therefore, $G'$ is isomorphic to one of the graphs $(C)$, $(D)$ and $(E)$ given in Figure \ref{fig:InducedTwoEdgesSubgraph}, where $C$ is a cycle of length five.

We now consider the case where the vertices of $xy$ and $wz$ have no common neighbor in $C$. For a cycle $C$ of length seven, by Corollary \ref{AtLeastTwoFour}, each of $S_{xy}$ and $S_{wz}$  contains at least four vertices. Then, the vertices of $xy$ and $wz$ must be adjacent to at least eight vertices in $C$, contradicting with the cardinality of $C$. Therefore, in this case there is no triangle-free EFC-graph where $C$ is a cycle of length seven. On the other hand, for a cycle $C$ of length five, each of $S_{xy}$ and $S_{wz}$  contains at least two vertices by Corollary \ref{AtLeastTwoFour}, see the graphs $(a)$, $(b)$ and $(c)$ in Figure \ref{fig:InducedOneEdgeSubgraph}. By the cardinality of $C$, the possible ordered pairs $(|S_{xy}|,|S_{wz}|)$ are $(2,2)$ and $(2,3)$. However, in the case with ordered pair $(2,2)$, one can easily verify that there exists an independent set $I$ of size three such that $G'\backslash I$ has a perfect matching, contradicting with Lemma \ref{independent-set}. Therefore, the only possible ordered pair is $(|S_{xy}|,|S_{wz}|)$ is $(2,3)$. Hence,  $G'$ is isomorphic to the graph $(F)$ given in Figure \ref{fig:InducedTwoEdgesSubgraph}, where $C$ is a cycle of length five.
\hfill\(\qed\)
\end{proof}

\begin{corollary}\label{covering-C}
Let $G$ be a connected triangle-free EFC-graph and $C$ be an induced odd cycle in $G$. The vertices of two edges from different components of $G\backslash C$ dominate the vertices of $C$. That is, for any two edges inducing a $2K_2$ in $G\backslash C$, each vertex of $C$ is adjacent to at least one vertex of these edges.
\end{corollary}

We conclude this section with the following result, which specifies the maximum number of components in $G\backslash C$. One can easily verify this result by combining Theorems \ref{C-OneEdge-Structure} and \ref{C-TwoEdge-Structure}.

\begin{corollary}\label{TwoComponents}
Let $G$ be a connected triangle-free EFC-graph and $C$ be an induced odd cycle in $G$. Then the subgraph $G\backslash C$ of $G$ has at most two components.
\end{corollary}

\section{Triangle-Free EFC-Graph Families}
The goal in this section is to obtain the families of triangle-free EFC-graphs. Recall that $G$ is a connected triangle-free EFC-graph with an induced odd cycle $C$ of length five or seven.
In the case where $G\backslash C$ is empty, $G$ is trivially $C_5$ or $C_7$.
Here, we analyze the non-trivial cases; that is, $G\backslash C$ is nonempty.
The number of components of $G \backslash C$ is denoted by $c(G\backslash C)$. By Corollary \ref{TwoComponents}, $G\backslash C$ contains at most two components; that is, $c(G\backslash C) \leq 2$. Theorem \ref{KnnProof}  states that every component of $G\backslash C$ is a randomly matchable triangle-free graph.
By Lemma \ref{InducedEFCSubgraph}, the vertices of $C$ and a component of $G\backslash C$ induce a connected triangle-free EFC-subgraph of $G$.
In Section \ref{oneComp},  we characterize connected triangle-free EFC graphs remaining connected after the removal of an induced odd cycle. In Section \ref{twoComp},  we provide a characterization of connected triangle-free EFC-graphs containing exactly two components after the removal of an induced odd cycle by using the characterization in Section \ref{oneComp}.

Let us now provide the definition of twin-free graphs in order to give explicit descriptions of  triangle-free EFC-graph families. Two vertices $u$ and $v$ are called \emph{twins} if they have the same set of neighbors; that is, $N(u) = N(v)$. In a set of twins, each pair of vertices are twins. A graph is called \emph{twin-free} if it does not have any twins.
For a given graph, a twin-free graph can be obtained by contracting each set $X$ of twins in the graph into a single vertex with multiplicity $|X|$.
In our graph family descriptions, we use the following notation: Let $H$ be a graph on $k$ vertices $v_1$, $v_2$,..., $v_k$ and let $m_1$, $m_2$,..., $m_k$ be nonnegative integers denoting the \emph{multiplicities} of these vertices, respectively. Then $H(m_1, m_2, ... ,m_k)$ denotes the graph obtained from $H$ by iteratively replacing each vertex $v_i$ with an independent set of $m_i$ vertices, each of which having the same neighborhood as $v_i$. Notice that all vertices in the independent set replacing $v_i$ are indeed twins. Clearly, $H = H(1, ... ,1)$ since all multiplicities are $1$.
To describe the structure of triangle-free EFC-graph families, we will use the twin-free graph $G^*$ given in Figure \ref{fig:MainTheorem}.

\begin{figure}[hbtp]
\centering

\begin{tikzpicture}[scale=1.5]

\filldraw [black] (-1,0) circle (1.5pt);
\filldraw [black] (1,0) circle (1.5pt);
\filldraw [black] (-1,-2) circle (1.5pt);
\filldraw [black] (1,-1) circle (1.5pt);
\filldraw [black] (-1,-1) circle (1.5pt);
\filldraw [black] (0,-2) circle (1.5pt);
\filldraw [black] (-2,0) circle (1.5pt);
\filldraw [black] (-2,-2) circle (1.5pt);
\filldraw [black] (2,-2) circle (1.5pt);
\filldraw [black] (2,0) circle (1.5pt);
\filldraw [black] (1,-2) circle (1.5pt);

\node[above] at (-2,0) {\footnotesize{$u_{11}$}};
\node[below] at (-2,-2) {\footnotesize{$u_{10}$}};
\node[below] at (-1,-2) {\footnotesize{$u_9$}};
\node[above] at (-1,0) {\footnotesize{$u_2$}};
\node[right] at (-1,-1) {\footnotesize{$u_1$}};
\node[below] at (0,-2) {\footnotesize{$u_5$}};
\node[above] at (1,0) {\footnotesize{$u_3$}};
\node[left] at (1,-1) {\footnotesize{$u_4$}};
\node[below] at (1,-2) {\footnotesize{$u_{6}$}};
\node[above] at (2,0) {\footnotesize{$u_{8}$}};
\node[below] at (2,-2) {\footnotesize{$u_{7}$}};

\draw[very thick] (-2,0) -- (2,0) -- (2,-2) -- (-2,-2) -- (-2,0);
\draw[very thick] (1,0) -- (1,-1);
\draw[very thick] (-1,0) -- (-1,-1);
\draw[ultra thick] (-2,-2) -- (-1,-1);
\draw[ultra thick] (2,-2) -- (1,-1);
\draw[ultra thick] (0,-2) -- (-1,-1);
\draw[ultra thick] (0,-2) -- (1,-1);
\end{tikzpicture}
\caption{Twin-free graph $G^*$ with vertices $u_1, u_2, u_3, u_4, u_5, u_6, u_7, u_8, u_9, u_{10}, u_{11}$.} \label{fig:MainTheorem}
\end{figure}

\subsection{Triangle-Free EFC Graphs with $c(G\backslash C)= 1$}\label{oneComp}

In this subsection, we study connected triangle-free EFC-graphs remaining connected after the removal of an induced odd cycle. We will show that the class of such EFC-graphs corresponds to the following graph class:

\begin{definition}\label{Fgraphs}
The graph class $\mathcal{F}$ is the union of the following graph families:
\begin{itemize}
  \item $\mathcal{F}_{11}  = \{G^*(1,1,1,1,1,n,n,0,0,0,0) \; | \; n\geq 1\}$
  \item $\mathcal{F}_{12}  = \{G^*(1,1,1,0,1,n+1,n+1,1,0,0,0) \;  | \;  n\geq 1\}$
  \item $\mathcal{F}_{21}  = \{G^*(1,1,1,n-r-s+1,1,r,n,s,0,0,0) \;  | \;  n\geq 1, n-1\geq r\geq 1, n-1\geq s \geq 1, n\geq r+s\}$
  \item $\mathcal{F}_{22}  = \{G^*(1,1,1,n-r-s,1,r+1,n+1,s+1,0,0,0) \;  | \;  n\geq 1, n-1\geq r\geq 1, n-1\geq s \geq 1, n \geq r+s\}$
  \item $\mathcal{F}_{3}   = \{G^*(1,1,r+1,s+1,1,0,n-s,n-r,0,0,0) \;  | \;  n\geq1, n-1\geq r\geq 1, n-1\geq s \geq 1\}$
  \item $\mathcal{F}_{4}   = \{G^*(r+1,n+1,s+1,1,1,0,0,0,0,0,n-r-s) \;  | \;  n\geq 1, n-1\geq r\geq 1, n-1\geq s \geq 1, n \geq r+s\}$
\end{itemize}
where $G^*$ is the twin-free graph depicted in Figure \ref{fig:MainTheorem}.
\end{definition}

The next result proves one direction as follows:
\begin{proposition}\label{easyDirection}
If $G \in \mathcal{F}$ where   $\mathcal{F} = \mathcal{F}_{11} \cup \mathcal{F}_{12} \cup \mathcal{F}_{21} \cup \mathcal{F}_{22} \cup \mathcal{F}_{3} \cup \mathcal{F}_{4}$, then $G$ is  a connected triangle-free EFC-graph remaining connected after the removal of an induced odd cycle.
\end{proposition}
\begin{proof}
For any $G \in \mathcal{F}$, one can check that there is no independent set $I$ of size three in $G$ such that $G \backslash I$ has a perfect matching. Hence, $G$ is an EFC-graph by Lemma \ref{independent-set}. All the other properties are easily verifiable.
\hfill\(\qed\)
\end{proof}

To prove the other direction, we suppose that $G\backslash C$ is connected; that is, $c(G\backslash C)= 1$. By Theorem \ref{KnnProof}, $G\backslash C$  is isomorphic to $K_{n,n}$ for some positive integer $n$. By Corollary \ref{AtLeastTwoFour}, the number of vertices in $C$ that are adjacent to the vertices of $G\backslash C$ is at least two or four, where $C$ is a cycle of length five or seven, respectively. Besides, all such vertices of $C$ are consecutive in either case of $C$. Hence, we observe that the number of vertices of $C$ that are nonadjacent to the vertices of $G\backslash C$ is at most three, and all such vertices are consecutive in either case of $C$. This observation allows us to analyze connected triangle-free EFC-graphs with ${c(G\backslash C)= 1}$ in the disjoint and complementary cases specified  in Propositions \ref{threeVertices}, \ref{twoVertices}, \ref{oneVertex}, \ref{zeroVertex}, and Lemma \ref{noc5c7}.

\begin{proposition}\label{threeVertices}
Let $G$ be a connected triangle-free EFC-graph having an induced odd cycle $C$ such that ${c(G\backslash C)= 1}$. If $C$ has exactly three  vertices that are nonadjacent to the vertices of $G\backslash C$, then $G \in \mathcal{F}_{11} \cup \mathcal{F}_{12}$. In particular, $G \in \mathcal{F}_{11}$ if $C$ is a cycle of length five and $G \in \mathcal{F}_{12}$ if $C$ is a cycle of length seven.
\end{proposition}
\begin{proof}
Let $G$ be a connected triangle-free EFC-graph having an induced odd cycle $C$ such that ${c(G\backslash C)= 1}$.
By Corollary \ref{c5c7}, $C$ is a cycle of length five or seven. By Theorem \ref{KnnProof}, $G\backslash C$ is isomorphic to $K_{n,n}$ for some positive integer $n$. Let the partite sets of $G\backslash C$ be labeled by $A$ and $B$. Let the consecutive  vertices of $C$ be labeled by $v_1$, $v_2$, $v_3$, $v_4$, $v_5$, if necessary $v_6$ and $v_7$, depending on the length of $C$. Suppose that three vertices of $C$ are nonadjacent to the vertices of $G\backslash C$. By Corollary \ref{AtLeastTwoFour}, these vertices of $C$ are consecutive in either case of $C$. W.l.o.g., let $v_1$, $v_2$, and $v_3$ be nonadjacent to the vertices of $G\backslash C$.

In the case where $C$ is a cycle of length five, by Theorem \ref{C-OneEdge-Structure}, the vertices $v_4$, $v_5$ and the endpoints of each edge in $G\backslash C$ induce the subgraph in Figure \ref{fig:InducedOneEdgeSubgraph} (a).
It follows that an endpoint of each edge in $G\backslash C$ is adjacent to $v_4$ and the other endpoint is adjacent to $v_5$.
Since $G$ is triangle-free, it follows that all vertices in a partite set of $G \backslash C$, say $A$, are adjacent to $v_4$ and all vertices in the other partite set $B$ are adjacent to $v_5$.
Then, the vertices $v_4$, $v_5$, and the vertices in $A$, $B$ induce a complete bipartite graph $K_{n+1,n+1}$.
Notice that any pair of vertices in $A$ are twins and any pair of vertices in $B$ are twins.
By representing the vertices $v_1$, $v_2$, $v_3$, $v_4$ and $v_5$ in $C$  with the vertices $u_1$, $u_2$, $u_3$, $u_4$ and $u_5$, respectively, and by representing the vertices in $A$ and $B$ with the vertices $u_7$ and $u_6$, respectively, in the graph $G^*$ given in Figure \ref{fig:MainTheorem}, it is easy to verify that $G$ belongs to the graph family $\mathcal{F}_{11}$.

Similarly, in the case where $C$ is a cycle of length seven, by Theorem \ref{C-OneEdge-Structure}, the vertices $v_4$, $v_5$, $v_6$, $v_7$ and the endpoints of each edge in $G\backslash C$ induce the subgraph in Figure \ref{fig:InducedOneEdgeSubgraph} (d).
It follows that an endpoint of each edge in $G\backslash C$ is adjacent to both $v_4$ and $v_6$  and the other endpoint is adjacent to  both $v_5$ and $v_7$.
Since $G$ is triangle-free, it follows that all vertices in a partite set of $G \backslash C$, say $A$, are adjacent to $v_4$ and $v_6$ and all vertices in the other partite set $B$ are adjacent to $v_5$ and $v_7$.
Hence, the vertices $v_4$, $v_5$, $v_6$, $v_7$, and the vertices in $A$, $B$ induce a complete bipartite graph with a missing edge, $K_{n+2,n+2} - v_4v_{7}$.
Any pair of the vertices in $A$ or $B$ are twins. Note that the vertex $v_5$ and any vertex in $A$ are also twins, and  the vertex $v_6$ and any vertex in $B$ are also twins.
By representing the vertices $v_1$, $v_2$, $v_3$, $v_4$, $v_5$, $v_6$, and $v_7$ in $C$ with $u_1$, $u_2$, $u_3$, $u_8$, $u_7$, $u_6$, and $u_5$, respectively, in the graph $G^*$ given in Figure \ref{fig:MainTheorem}, we conclude that $G$ belongs to the graph family $\mathcal{F}_{12}$.
\hfill\(\qed\)
\end{proof}

\begin{proposition}\label{twoVertices}
Let $G$ be a connected triangle-free EFC-graph having an induced odd cycle $C$ such that ${c(G\backslash C)= 1}$. If $C$ has exactly two  vertices that are nonadjacent to the vertices of $G\backslash C$, then $G \in \mathcal{F}_{21} \cup \mathcal{F}_{22}$. In particular, $G \in \mathcal{F}_{21}$ if $C$ is a cycle of length five and $G \in \mathcal{F}_{22}$ if $C$ is a cycle of length seven.
\end{proposition}
\begin{proof}
Let $G$ be a connected triangle-free EFC-graph having an induced odd cycle $C$ such that ${c(G\backslash C)= 1}$.
By Corollary \ref{c5c7}, $C$ is a cycle of length five or seven. By Theorem \ref{KnnProof}, $G\backslash C$ is isomorphic to $K_{n,n}$ for some positive integer $n$. Let the partite sets of $G\backslash C$ be labeled by $A$ and $B$. Let the consecutive  vertices of $C$ be labeled by $v_1$, $v_2$, $v_3$, $v_4$, $v_5$, if necessary $v_6$ and $v_7$, depending on the length of $C$. Suppose that two vertices of $C$ are nonadjacent to the vertices of $G\backslash C$. By Corollary \ref{AtLeastTwoFour}, these vertices of $C$ are consecutive  in either case of $C$. W.l.o.g., let $v_1$ and $v_2$ be nonadjacent to the vertices of $G\backslash C$.

Let $C$ be a cycle on five vertices. By Theorem \ref{C-OneEdge-Structure},
the vertices $v_3$, $v_4$, $v_5$ and the endpoints of each edge in $G\backslash C$ induce the subgraph in Figure \ref{fig:InducedOneEdgeSubgraph} (a) or (b). It follows that an endpoint of each edge in $G\backslash C$ is adjacent to the vertex $v_4$ whereas the other endpoint is adjacent to at least one of $v_3$ and $v_5$. Particularly, if an endpoint of an edge in  $G\backslash C$ is adjacent to only one of $v_3$ and $v_5$, it is the graph in Figure \ref{fig:InducedOneEdgeSubgraph} (a); and if  an endpoint of an edge in  $G\backslash C$ is adjacent to both $v_3$ and $v_5$, it is the graph in Figure \ref{fig:InducedOneEdgeSubgraph} (b). Hence, the triangle-freeness of $G$ implies that all vertices in a partite set of $G \backslash C$, say $A$, are adjacent to $v_4$ and all vertices in the other partite set $B$ are adjacent to at least one of $v_3$ and $v_5$.
Note here that the vertex $v_4$ and the vertices in $A$ and $B$ induce a complete bipartite graph $K_{n+1,n}$;
and the vertex $v_4$ and any vertex in $B$ that is adjacent to both $v_3$ and $v_5$ are twins.
Therefore, by representing the vertices $v_1$, $v_2$, $v_3$, $v_4$, and $v_5$ with the vertices $u_1$, $u_2$, $u_3$, $u_4$, and $u_5$, respectively, and the vertices in $A$ with the vertex $u_7$ and by representing the vertices in $B$ which are adjacent to only $v_3$ with the vertex $u_8$ and are adjacent to only $v_5$ with the vertex $u_6$ in the graph $G^*$ given in Figure \ref{fig:MainTheorem}, one can observe that $G$ belongs to the graph family $\mathcal{F}_{21}$.

Now let $C$ be a cycle on seven vertices. By Theorem \ref{C-OneEdge-Structure},
the vertices $v_3$, $v_4$, $v_5$, $v_6$, $v_7$ and the endpoints of each edge in $G\backslash C$ induce the subgraph in Figure \ref{fig:InducedOneEdgeSubgraph} (d) or (e). It follows that an endpoint of each edge in $G\backslash C$ is adjacent to the vertex $v_5$ and at least one of the vertices $v_3$ and $v_7$ whereas the other endpoint is adjacent to both $v_4$ and $v_6$. Particularly, if an endpoint of an edge in  $G\backslash C$ is adjacent to only one of $v_3$ and $v_7$, it is the graph in Figure \ref{fig:InducedOneEdgeSubgraph} (d); and if an endpoint of an edge in  $G\backslash C$  is adjacent to both $v_3$ and $v_7$, it is the graph in Figure \ref{fig:InducedOneEdgeSubgraph} (e). Since $G$ is triangle-free, all vertices in a partite set of $G\backslash C$, say $A$, are adjacent to the vertices $v_4$ and $v_6$, and all vertices in the other partite set $B$ are adjacent to the vertex $v_5$ and are adjacent to at least one of the vertices $v_3$ and $v_7$. It follows that the vertices $v_4$, $v_5$, $v_6$, and the vertices in $A$ and $B$ induce a complete bipartite graph $K_{n+2,n+1}$. Note here that the vertex $v_5$ and any vertex in $A$ are twins. Furthermore, the vertex $v_4$ and any vertex in $B$ being adjacent to $v_3$ and $v_5$ but not to $v_7$  are twins, and the vertex $v_6$ and any vertex in $B$ being adjacent to $v_5$ and $v_7$ but not to $v_3$ are twins.
Hence, by representing the vertices $v_1$, $v_2$, $v_3$, $v_4$, $v_5$, $v_6$, and $v_7$ with the vertices $u_1$, $u_2$, $u_3$ $u_8$, $u_7$, $u_6$, and $u_5$, respectively, and by representing the vertices in $B$ which are adjacent to both $v_3$ and $v_7$ with the vertex $u_4$ in graph $G^*$ given in Figure \ref{fig:MainTheorem}, it is easy to verify that $G$ belongs to the graph family $\mathcal{F}_{22}$.
\hfill\(\qed\)
\end{proof}

We now consider the case where $C$ has at most one vertex that is nonadjacent to the vertices of $G\backslash C$.
In the next result, we show that there is no triangle-free EFC-graph $G$ in such a case where $C$ is a cycle of length seven. In Propositions \ref{oneVertex} and \ref{zeroVertex}, we prove that triangle-free EFC-graph $G$ belongs to $\mathcal{F}_{3}$ or $\mathcal{F}_{4}$ in such a case where $C$ is a cycle of length five.

\begin{lemma}\label{noc5c7}
Let $G$ be a connected triangle-free EFC-graph having an induced odd cycle $C$ of length seven such that ${c(G\backslash C)= 1}$. Then, at least two (consecutive) vertices of $C$ are nonadjacent to the vertices of $G\backslash C$.
\end{lemma}

\begin{proof}
Let $G$ be a connected triangle-free EFC-graph having an induced odd cycle $C$ of length seven such that ${c(G\backslash C)= 1}$. By Theorem \ref{KnnProof}, $G\backslash C$ is isomorphic to $K_{n,n}$ for some positive integer $n$. Let  the consecutive  vertices of $C$ be labeled by $v_1$, $v_2$, $v_3$, $v_4$, $v_5$, $v_6$, and $v_7$. Let $xy$ be an edge in $G \backslash C$. Note that the vertices $x$ and $y$ belong to the different partite sets of $G\backslash C$. By Theorem \ref{C-OneEdge-Structure} and Corollary \ref{AtLeastTwoFour}, $N(x) \cup N(y)$ contains at least four consecutive vertices, say $v_1$, $v_2$, $v_3$, and $v_4$. Without loss of generality, let $x$ be adjacent to $v_1$ and $v_3$, and $y$ be adjacent to $v_2$ and $v_4$.

We first show that the vertex $v_6$ is nonadjacent to the vertices of $G\backslash C$.
Conversely, we suppose that the vertex $v_6$ is adjacent to a vertex $w \in G\backslash C$. Let $w$ and $x$ be in the same partite set of $G\backslash C$; that is, $wy \in E(G\backslash C)$. It follows that the vertices $v_3$, $v_5$, and $v_7$ induce an independent set of size $3$ in $G$ and $G \backslash \{v_3, v_5, v_7\}$ has a perfect matching containing the edges $v_1v_2$, $v_4y$, $v_6w$, and a perfect matching of randomly matchable graph $(G\backslash C) \backslash \{w,y\}$ which is isomorphic to $K_{n-1,n-1}$. It gives a contradiction with Lemma \ref{independent-set}. Hence, we conclude that the vertex $v_6$ is adjacent to none of the vertices of $G\backslash C$.

We now show that at least one of the vertices $v_5$ and $v_7$ is nonadjacent to the vertices of $G\backslash C$.
Assume to the contrary that there exist two vertices $a$ and $b$ in $G\backslash C$ such that the vertices $v_5$ and $v_7$ are adjacent to $a$ and $b$, respectively. Notice that $a$ and $b$ cannot belong to the different partite sets of $G\backslash C$; that is, $ab \notin E(G\backslash C)$. Otherwise, by contradicting with Lemma \ref{independent-set}, the vertices $v_1$, $v_4$, and $v_6$ induce an independent set of size $3$ in $G$, and $G \backslash \{v_1, v_4, v_6\}$ has a perfect matching containing the edges $v_2v_3$, $v_5a$, $v_7b$, and a perfect matching of randomly matchable graph $(G\backslash C) \backslash \{a,b\}$ which is isomorphic to $K_{n-1,n-1}$. Thus, the vertices $a$ and $b$ belong to the same partite set of $G\backslash C$.
Note that the vertices $x$ and $y$ belong to the different partite sets of $G\backslash C$.
In the case where $a$, $b$ and $x$ belong to the same partite set of  $G\backslash C$, we observe that the vertices $v_1$, $v_3$, and $v_6$ induce an independent set of size $3$ in $G$. Besides, $G \backslash \{v_1, v_3, v_6\}$ has a perfect matching containing the edges $v_2y$, $v_4v_5$, $v_7b$, and a perfect matching of randomly matchable graph $(G\backslash C) \backslash \{b,y\}$, by contradicting with Lemma \ref{independent-set}. Since $G$ is triangle-free, it is clear that $b$ and $x$ cannot refer to the same vertex. The same result holds even if  $a$ and $x$ or $a$ and $b$ refer to the same vertex.
Similarly, in the case where $a$, $b$ and $y$ belong to the same partite set of  $G\backslash C$, the vertices $v_2$, $v_4$, and $v_6$ induce an independent set of size $3$ in $G$. Since $G \backslash \{v_2, v_4, v_6\}$ has a perfect matching containing the edges $v_1v_7$, $v_5a$, $v_3x$, and a perfect matching of randomly matchable graph $(G\backslash C) \backslash \{a,x\}$, it contradicts with Lemma \ref{independent-set}. The vertices $a$ and $y$ cannot refer the same vertex; and the same result holds even if  $b$ and $y$ or $a$ and $b$ refer to the same vertex. Thus, we came to the conclusion that there are no such $a$ and $b$ vertices in $G\backslash C$. Therefore, the proof is complete.
      \hfill\(\qed\)
\end{proof}

\begin{proposition}\label{oneVertex}
Let $G$ be a connected triangle-free EFC-graph having an induced odd cycle $C$ such that ${c(G\backslash C)= 1}$. If $C$ has exactly one vertex that is nonadjacent to the vertices of $G\backslash C$, then $G \in \mathcal{F}_{3}$. In particular, $C$ is a cycle of length five.
\end{proposition}
\begin{proof}
Let $G$ be a connected triangle-free EFC-graph having an induced odd cycle $C$  such that ${c(G\backslash C)= 1}$.
Suppose that one vertex of $C$ is nonadjacent to the vertices of $G\backslash C$. By Lemma \ref{noc5c7}, $C$ is a cycle of length five.
By Theorem \ref{KnnProof}, $G\backslash C$ is isomorphic to $K_{n,n}$ for some positive integer $n$. Let the partite sets of $G\backslash C$ be labeled as $A$ and $B$, and let the consecutive  vertices of $C$ be labeled as $v_1$, $v_2$, $v_3$, $v_4$, and $v_5$. W.l.o.g, let $v_1$ be adjacent to none of the vertices of $G\backslash C$.
Then, there exist the vertices $a$ and $b$ in $G\backslash C$ such that the vertices $v_2$ and $v_5$ are adjacent to $a$ and $b$, respectively.

By  Theorem \ref{C-OneEdge-Structure}, each neighbor of $a$ in $G\backslash C$ is adjacent to the vertex $v_3$ and each neighbor of $b$ in $G\backslash C$ is adjacent to the vertex $v_4$. By the triangle-freeness of $G$, it is easy to see that $a$ and $b$ has no common neighbor in $G\backslash C$. Since $G\backslash C$ is isomorphic to $K_{n,n}$, we conclude that the vertices $a$ and $b$ belong to the different partite sets of $G \backslash C$, say $a \in A$ and $b \in B$. It follows that each vertex of $A$ in $G\backslash C$ is adjacent to the vertex $v_4$ and  each vertex of $B$ in $G\backslash C$ is adjacent to the vertex $v_3$. Hence, the vertices $v_3$, $v_4$ and the vertices in $A$ and $B$ induce a complete bipartite graph $K_{n+1,n+1}$. Note here that the vertex $v_3$ and any vertex in $A$ being adjacent to both $v_2$ and $v_4$ are twins, and the vertex $v_4$ and any vertex in $B$ being adjacent to both $v_3$ and $v_5$ are twins.
Hence, by representing vertices $v_1$, $v_2$, $v_3$, $v_4$, and $v_5$ with the vertices  $u_1$, $u_2$ ,$u_3$, $u_4$, and $u_5$, respectively, and by representing the vertices in $A$ being nonadjacent to $v_2$ with the vertex $u_8$ and the vertices in $B$ being nonadjacent to $v_5$ with the vertex $u_{11}$ in graph $G^*$ given in Figure \ref{fig:MainTheorem}, it is easy to verify that $G$ belongs to the graph family $\mathcal{F}_{3}$.
\hfill\(\qed\)
\end{proof}

Suppose that three vertices of $C$ are nonadjacent to the vertices of $G\backslash C$

\begin{proposition}\label{zeroVertex}
Let $G$ be a connected triangle-free EFC-graph having an induced odd cycle $C$ such that ${c(G\backslash C)= 1}$. If each vertex in $C$ is adjacent to at least one vertex of $G\backslash C$, then $G \in \mathcal{F}_{4}$. In particular, $C$ is a cycle of length five.
\end{proposition}
\begin{proof}
Let $G$ be a connected triangle-free EFC-graph having an induced odd cycle $C$ such that ${c(G\backslash C)= 1}$.
Suppose that there is no vertex in $C$ being adjacent to none of the vertices of $G\backslash C$; that is, each vertex in $C$ is adjacent to at least one vertex of $G\backslash C$. Then, by Lemma \ref{noc5c7}, $C$ is a cycle of length five.
By Theorem \ref{KnnProof}, $G\backslash C$ is isomorphic to $K_{n,n}$ for some positive integer $n$. Let the partite sets of $G\backslash C$ be labeled by $A$ and $B$, and the consecutive  vertices of $C$ be labeled by $v_1$, $v_2$, $v_3$, $v_4$, and $v_5$.  By the triangle-freeness of $G$, it is easy to see that all neighbors of each vertex in $C$ belong to the same partite set of $G\backslash C$.

Let $xy$ be an edge in $G\backslash C$, where $x \in A$ and $y \in B$. By Lemma \ref{EdgeEdge}, $xy$ induces a $C_4$ with at least one edge of $C$, say $v_1v_5$ with edges $xv_1$ and $yv_5$. It follows that $N_{G\backslash C}(v_1) \subseteq A$ and $N_{G\backslash C}(v_5) \subseteq B$. Since $y \in B$, by Theorem \ref{C-OneEdge-Structure}, each vertex in $A$ is adjacent to at least one of $v_1$ and $v_4$. Since $x \in A$, by Theorem \ref{C-OneEdge-Structure}, each vertex in $B$ is adjacent to at least one of $v_2$ and $v_5$. We now consider the neighbors of $v_3$ in $G\backslash C$. Since all neighbors of $v_3$ belong to the same partite set of $G\backslash C$, without loss of generality, we suppose that $N_{G\backslash C}(v_3) \subseteq A$. Thus, there exists a vertex $w \in A$ which is adjacent to $v_3$. Then, Theorem \ref{C-OneEdge-Structure} implies that each vertex in $B$ is adjacent to at least one of $v_2$ and $v_4$.

At this point, we have that each vertex in $B$ is adjacent to not only at least one of $v_2$ and $v_5$ but also at least one of $v_2$ and $v_4$. Since a vertex in $B$ cannot be adjacent to both $v_4$ and $v_5$ by the triangle-freeness of $G$, it is easy to observe that each vertex in $B$ is adjacent to the vertex $v_2$. Particularly, $y$ is adjacent to both $v_2$ and $v_5$. Then, by Theorem \ref{C-OneEdge-Structure}, it follows that each vertex in $A$ is adjacent to at least one of $v_1$ and $v_3$. Remark that each vertex in $A$ is also adjacent to at least one of $v_1$ and $v_4$. Since a vertex in $A$ cannot be adjacent to both $v_3$ and $v_4$ by the triangle-freeness of $G$, each vertex in $A$ is indeed adjacent to the vertex $v_1$. Hence, we make a result that the vertices $v_1$, $v_2$ and the vertices in $A$ and $B$ induce a complete bipartite graph $K_{n+1,n+1}$ in $G$. Remark that $N_{G\backslash C}(v_3) \subseteq A$ and $N_{G\backslash C}(v_5) \subseteq B$.

We now consider the neighbors of $v_4$ in $G\backslash C$. W.l.o.g., we suppose that $N_{G\backslash C}(v_4) \subseteq B$; that is, there exists a vertex $z \in B$ which is adjacent to $v_4$. By Theorem \ref{C-OneEdge-Structure}, each vertex in $A$ is adjacent to at least one of $v_3$ and $v_5$. Since $N_{G\backslash C}(v_5) \subseteq B$, it is easy to see that each vertex in $A$ is adjacent to the vertex $v_3$. Hence, we conclude that the subgraph of $G$ induced by the vertices $v_1$, $v_2$, $v_3$ and the vertices in $A$ and $B$ is a complete bipartite graph $K_{n+2,n+1}$. Note here that the vertex $v_2$ and any vertex in $A$ are twins. Furthermore, the vertex $v_3$ and any vertex in $B$ being adjacent to both $v_2$ and $v_4$ are twins; and the vertex $v_1$ and any vertex in $B$ being adjacent to $v_2$ and $v_5$ are twins. Notice also that $N_{B}(v_4) \cap N_{B}(v_5) = \emptyset$ and there may exist some vertices in $B$ which are adjacent to neither $v_4$ nor $v_5$.
Therefore, by representing the vertices $v_1$, $v_2$, $v_3$ ,$v_4$, and $v_5$ with the vertices $u_1$, $u_2$, $u_3$, $u_4$, and $u_5$, respectively, and by representing the vertices in $B$ being adjacent to none of $v_4$ and $v_5$ with the vertex $u_{11}$, in graph $G^*$ given in Figure \ref{fig:MainTheorem}, we conclude that $G$ belongs to the graph family $\mathcal{F}_{4}$.
\hfill\(\qed\)
\end{proof}

We summarize the results in this subsection as follows:

\begin{theorem}\label{main1}
  A graph $G$ is a connected triangle-free EFC-graph having an induced odd cycle $C$ such that ${c(G\backslash C)= 1}$ if and only if $G \in \mathcal{F}$, where $\mathcal{F} = \mathcal{F}_{11} \cup \mathcal{F}_{12} \cup \mathcal{F}_{21} \cup \mathcal{F}_{22} \cup \mathcal{F}_{3} \cup \mathcal{F}_{4}$.
\end{theorem}
\begin{proof}
  One direction easily follows from Proposition \ref{easyDirection}. We proceed as follows to prove the other direction: Let $G$ be a connected triangle-free EFC-graph having an induced odd cycle $C$ such that ${c(G\backslash C)= 1}$. Let $S$ be the set of vertices in $C$ which are nonadjacent to the vertices of $G\backslash C$. By Corollary \ref{AtLeastTwoFour}, we observe that $|S| \leq 3$. If $|S| = 3$, then by Proposition \ref{threeVertices}, $G$ is a member of $\mathcal{F}_{11} \cup \mathcal{F}_{12}$. If $|S| = 2$, then by Proposition \ref{twoVertices}, $G$ is a member of $\mathcal{F}_{21} \cup \mathcal{F}_{22}$.  If $C$ is a cycle of length seven, then by Lemma \ref{noc5c7}, there is no such graph $G$ in the case $|S| \leq 1$. Finally, if $C$ is a cycle of length five, then by Propositions \ref{oneVertex} and \ref{zeroVertex}, $G$ is a member of $\mathcal{F}_{3} \cup \mathcal{F}_{4}$ in the case $|S| \leq 1$.
  \hfill\(\qed\)
\end{proof}

\subsection{Triangle-Free EFC Graphs with $c(G\backslash C)= 2$}\label{twoComp}

In this subsection, we study  connected triangle-free EFC-graphs containing two components after the removal of an induced odd cycle. We will show that the class of such EFC-graphs corresponds to the following graph class:

\begin{definition}\label{Ggraphs}
The graph class $\mathcal{G}$ contains the following graph families:
\begin{itemize}
  \item $\mathcal{G}_{11}  = \{G^*(m+1,m+1,1,0,1,1,n+1,n+1,0,0,0)  \; | \; n\geq1, m\geq1\}$
  \item $\mathcal{G}_{12}  = \{G^*(m+1,m+1,1,n-r-s,1,r+1,n+1,s+1,0,0,0)  \; | \; n\geq 1, m\geq 1, n-1\geq r\geq 1, n-1\geq s \geq 1, n \geq r+s\}$
  \item $\mathcal{G}_{21}  = \{G^*(1,1,1,n-r-s+1,1,r,n,s,0,m,m)  \; | \; n\geq 1, m\geq 1, n-1\geq r\geq 1, n-1\geq s \geq 1, n \geq r+s\}$.
  \item $\mathcal{G}_{22}  = \{G^*(1,1,r+1,s+1,1,0,n-s,n-r,0,m,m)  \; | \; n\geq1,  m\geq 1, n-1\geq r\geq 1, n-1\geq s \geq 1\}$

  \item $\mathcal{G}_{23}  = \{G^*(r+1,n+1,s+1,1,1,m,m,0,0,0,n-r-s)  \; | \; n\geq 1, m\geq 1, n-1\geq r\geq 1, n-1\geq s \geq 1, n \geq r+s\}$

  \item $\mathcal{G}_{31}  = \{G^*(m-k-l+1,1,1,n-r-s+1,1,r,n,s,l,m,k)  \; | \; n\geq 1, m\geq 1, n-1\geq r\geq 1, n-1\geq s \geq 1, n\geq r+s, m-1\geq l\geq 1, m-1\geq k \geq 1, m\geq k+l\}$

  \item $\mathcal{G}_{32}  =  \{G^*(k+1,l+1,1,n-r-s+1,1,r,n,s,0,m-l,m-k)  \; | \; n\geq 1, m\geq 1, n-1\geq r\geq 1, n-1\geq s \geq 1, n\geq r+s, m-1\geq l\geq 1, m-1\geq k \geq 1, m\geq k+l\}$

\end{itemize}
where $G^*$ is the twin-free graph depicted in Figure \ref{fig:MainTheorem}.
\end{definition}

In the next result, we prove one direction as follows:

\begin{proposition}\label{easyDirection2}
If $G \in \mathcal{G}$ where   $\mathcal{G} = \mathcal{G}_{11} \cup \mathcal{G}_{12} \cup \mathcal{G}_{21} \cup \mathcal{G}_{22} \cup \mathcal{G}_{23} \cup \mathcal{G}_{31} \cup \mathcal{G}_{32}$, then $G$ is  a connected triangle-free EFC-graph containing exactly two components after the removal of an induced odd cycle.
\end{proposition}
\begin{proof}
As in Proposition \ref{easyDirection}, one can verify that there is no independent set $I$ of size three in $G \in \mathcal{G}$ such that $G \backslash I$ has a perfect matching; hence by Lemma \ref{independent-set},  $G$ is an EFC-graph. All the other properties are easily verifiable.
\hfill\(\qed\)
\end{proof}

To prove the other direction, we suppose that $c(G\backslash C)= 2$; that is, $G\backslash C$ contains two components, say $H_1$ and $H_2$. By Theorem \ref{KnnProof}, $H_1$ and $H_2$ are isomorphic to $K_{n,n}$ and $K_{m,m}$, respectively, for some nonnegative integers $n$ and $m$. By Lemma \ref{InducedEFCSubgraph}, both subgraphs of $G$ induced by the vertices of $C$ and $H_1$ and by the vertices of $C$ and $H_2$, i.e., $G\backslash H_2$ and $G\backslash H_1$, respectively, are connected triangle-free EFC-graphs. Besides, $(G\backslash H_2)\backslash C$ and $(G\backslash H_1)\backslash C$ are connected; indeed, $(G\backslash H_2)\backslash C = H_1$ and $(G\backslash H_1)\backslash C= H_2$ which are complete bipartite graphs $K_{n,n}$ and $K_{m,m}$, respectively.
Hence, by Theorem \ref{main1}, we deduce that $G\backslash H_2$ and $G\backslash H_1$ are members of the graph class $\mathcal{F}$ defined in Section \ref{oneComp}.
On the other hand, by Lemma \ref{DisjointEdges1}, any two edges from $H_1$ and $H_2$ cannot induce a $C_4$ with the same edge in $C$, and by Corollary \ref{covering-C}, the vertices of these edges dominate the vertices of $C$. Furthermore,  Corollary \ref{only-one-odd-cycle} states that there are no two vertex-disjoint induced odd cycles in $G$. These observations allow us to reduce the rest of this subsection to merging two graphs from $\mathcal{F}$ with each other.
Note here that the merged graphs should have the same induced odd cycle $C$ and the components of $G\backslash C$ have at most two nonconsecutive common neighbors in $C$.

In the next result, we provide the characterization of connected triangle-free EFC-graphs with ${c(G\backslash C)= 2}$ where $C$ is a cycle of length seven.

\begin{proposition}\label{c7-c7}
Let $G$ be a connected triangle-free EFC-graph having an induced odd cycle $C$ of length seven such that ${c(G\backslash C)= 2}$. Then $G \in \mathcal{G}_{11} \cup  \mathcal{G}_{12}$.
\end{proposition}
\begin{proof}
Let $G$ be a connected triangle-free EFC-graph having an induced odd cycle $C$ of length seven such that ${c(G\backslash C)= 2}$. Let $H_1$ and $H_2$ be the components of $c(G\backslash C)$.  By Lemma \ref{InducedEFCSubgraph} and by Theorem \ref{main1}, induced subgraphs $G\backslash H_2$ and $G\backslash H_1$ of $G$ are both contained in the graph class $\mathcal{F}$. By Propositions \ref{threeVertices} and \ref{twoVertices}, each of $G\backslash H_2$ and $G\backslash H_1$ of $G$ is a member of the union of the graph families $\mathcal{F}_{12}$ and $\mathcal{F}_{22}$; that is, $ G\backslash H_1, G\backslash H_2 \in \mathcal{F}_{12} \cup \mathcal{F}_{22}$. However, by Lemma \ref{DisjointEdges1}, the subgraphs $G\backslash H_2$ and $G\backslash H_1$ cannot both belong to the graph family $\mathcal{F}_{22}$. It follows that at least one of $G\backslash H_2$ and $G\backslash H_1$ is a member of the graph family $\mathcal{F}_{12}$.
If $G\backslash H_2$ and $G\backslash H_1$ both belong to the graph family $\mathcal{F}_{12}$, the endpoints of any two edges from $H_1$ and $H_2$ and vertices of $C$ induce the graph in Figure \ref{fig:InducedTwoEdgesSubgraph} (A).
Then, by representing the cycle $C$ with $u_1u_2u_3u_8u_7u_6u_5$ in graph $G^*$ given in Figure \ref{fig:MainTheorem}, it is easy to see that $G$ belongs to the graph family $\mathcal{G}_{11}$.

Suppose that $G\backslash H_2$ belongs to the graph family $\mathcal{F}_{12}$ and $G\backslash H_1$  belongs to the graph family $\mathcal{F}_{22}$. Since $G\backslash H_1 \in \mathcal{F}_{22}$, there exists an edge $xy$ in $H_2$ such that the vertices $x$, $y$ and the vertices of $C$ induce the graph in Figure \ref{fig:InducedOneEdgeSubgraph} (e). It follows that
the vertices $x$, $y$, the vertices of $C$ and the endpoints of any edge of $H_1$ induce the graph in Figure \ref{fig:InducedTwoEdgesSubgraph} (B). Let $x$ be the vertex being adjacent to exactly three vertices of $C$.
Hence, by representing the cycle $C$ with $u_1u_2u_3u_8u_7u_6u_5$ and by representing $x$ and its twins with the vertex $u_4$ in graph $G^*$ given in Figure \ref{fig:MainTheorem}, we conclude that $G$ belongs to the graph family $\mathcal{G}_{12}$.
\hfill\(\qed\)
\end{proof}

We now proceed with the characterization of connected triangle-free EFC-graphs with ${c(G\backslash C)= 2}$ where $C$ is a cycle of length five. In the next lemma, we show that at least one of $G\backslash H_2$ and $G\backslash H_1$ of $G$ is a graph contained in the graph family $\mathcal{F}_{11}$ or $\mathcal{F}_{21}$.

\begin{lemma}\label{noC5}
Let $G$ be a connected triangle-free EFC-graph having an induced odd cycle $C$ of length five such that ${c(G\backslash C)= 2}$. If $H_1$ and $H_2$ are the components of $c(G\backslash C)$, then at least one of the induced subgraphs $G\backslash H_2$ and $G\backslash H_1$ belongs to the union of the graph families $\mathcal{F}_{11}$ and $\mathcal{F}_{21}$.
\end{lemma}
\begin{proof}
Let $G$ be a connected triangle-free EFC-graph having an induced odd cycle $C$ of length five such that ${c(G\backslash C)= 2}$. Let $H_1$ and $H_2$ be the components of $c(G\backslash C)$.  By Lemma \ref{InducedEFCSubgraph} and Theorem \ref{main1}, induced subgraphs $G\backslash H_2$ and $G\backslash H_1$ of $G$ are both contained in the graph class $\mathcal{F}$. Since $C$ is a cycle of length five, each of $G\backslash H_2$ and $G\backslash H_1$ belongs to one of the graph families $\mathcal{F}_{11}$, $\mathcal{F}_{21}$, $\mathcal{F}_{3}$, and $\mathcal{F}_{4}$.
Notice that $(G\backslash H_2)\backslash C = H_1$ and $(G\backslash H_1)\backslash C = H_2$.
By Lemma \ref{DisjointEdges1} and  Corollary \ref{covering-C}, the vertices of $H_1$ and $H_2$ have at most two nonconsecutive common neighbors in $C$. By definition of the graph families  $\mathcal{F}_{3}$ and $\mathcal{F}_{4}$, we observe that $G\backslash H_2$ and $G\backslash H_1$ cannot both belong to $\mathcal{F}_{3} \cup \mathcal{F}_{4}$. Therefore, it follows that at least one of $G\backslash H_2$ and $G\backslash H_1$ belongs to  the graph families $\mathcal{F}_{11}$ and $\mathcal{F}_{21}$.
  \hfill\(\qed\)
\end{proof}

\begin{proposition}\label{c5-c5}
Let $G$ be a connected triangle-free EFC-graph having an induced odd cycle $C$ of length five such that ${c(G\backslash C)}$ has the components $H_1$ and $H_2$. Then $G \in \mathcal{G}_{21} \cup  \mathcal{G}_{22} \cup \mathcal{G}_{23} \cup \mathcal{G}_{31} \cup  \mathcal{G}_{32}$.
\end{proposition}
\begin{proof}
Let $G$ be a connected triangle-free EFC-graph having an induced odd cycle $C$ of length five such that ${c(G\backslash C)= 2}$. Let $H_1$ and $H_2$ be the components of $c(G\backslash C)$. By Lemma \ref{InducedEFCSubgraph} and Theorem \ref{main1}, induced subgraphs $G\backslash H_2$ and $G\backslash H_1$ of $G$ are both contained in the graph class $\mathcal{F}$. Since $C$ is a cycle of length five, each of $G\backslash H_2$ and $G\backslash H_1$ belongs to the union of the graph families $\mathcal{F}_{11}$, $\mathcal{F}_{21}$, $\mathcal{F}_{3}$, and $\mathcal{F}_{4}$. By Lemma \ref{noC5}, at least one of the induced subgraphs $G\backslash H_2$ and $G\backslash H_1$ belongs to the union of the graph families $\mathcal{F}_{11}$ and $\mathcal{F}_{21}$.

We first suppose that one of $G\backslash H_2$ and  $G\backslash H_1$ belongs to the graph family $\mathcal{F}_{11}$; say  $G\backslash H_2 \in \mathcal{F}_{11}$. Then, by definition of $\mathcal{F}_{11}$, each edge of $H_1$ induce a $C_4$ with the same edge of $C$.  Since, by Corollary \ref{covering-C}, any two edges from $H_1$ and $H_2$ dominate the vertices of $C$, it is easy to see that $G\backslash H_1 \notin \mathcal{F}_{11}$. It follows that $G\backslash H_1$ is contained in one of the graph families $\mathcal{F}_{21}$, $\mathcal{F}_{3}$ and $\mathcal{F}_{4}$.
Hence, by merging a graph from $\mathcal{F}_{11}$ and a graph from $\mathcal{F}_{21}$ via common cycle $C$ represented with $u_1u_2u_3u_4u_5u_1$  in graph $G^*$ given in Figure \ref{fig:MainTheorem}, we conclude that $G$ belongs to the graph family $\mathcal{G}_{21}$. Similarly, by merging a graph from $\mathcal{F}_{11}$ and a graph from $\mathcal{F}_{3}$ or $\mathcal{F}_{4}$ via common cycle $C$ represented with $u_1u_2u_3u_4u_5u_1$  in graph $G^*$ given in Figure \ref{fig:MainTheorem}, we conclude that $G$ belongs to the graph family $\mathcal{G}_{22}$ or $\mathcal{G}_{23}$, respectively.

We now suppose that one of $G\backslash H_2$ and  $G\backslash H_1$ belongs to the graph family $\mathcal{F}_{21}$; say  $G\backslash H_2 \in \mathcal{F}_{21}$. In the case where $G\backslash H_1 \in \mathcal{F}_{11}$; we previously obtain the graph family $\mathcal{G}_{21}$. Then, we suppose also that $G\backslash H_1 \notin \mathcal{F}_{11}$. Note that $(G\backslash H_2)\backslash C = H_1$ and $(G\backslash H_1)\backslash C = H_2$. By Lemma \ref{DisjointEdges1} and  Corollary \ref{covering-C}, the vertices of $H_1$ and $H_2$ have at most two nonconsecutive common neighbors in $C$.
By definition of the graph families $\mathcal{F}_{4}$, it follows that $G\backslash H_1 \notin \mathcal{F}_{4}$. Therefore,
by merging a graph from $\mathcal{F}_{21}$ and a graph from $\mathcal{F}_{21}$ or $\mathcal{F}_{3}$ via common cycle $C$ represented with $u_1u_2u_3u_4u_5u_1$  in graph $G^*$ given in Figure \ref{fig:MainTheorem}, it is easy to observe that $G$ belongs to the graph family $\mathcal{G}_{31}$ or $\mathcal{G}_{32}$.
\hfill\(\qed\)
\end{proof}

We summarize the results in this subsection as follows:
\begin{theorem}\label{main2}
  A graph $G$ is a connected triangle-free EFC-graph having an induced odd cycle $C$  such that ${c(G\backslash C)= 2}$ if and only if $G \in \mathcal{G}$ where $\mathcal{G} = \mathcal{G}_{11} \cup \mathcal{G}_{12} \cup \mathcal{G}_{21} \cup \mathcal{G}_{22} \cup \mathcal{G}_{23} \cup \mathcal{G}_{31} \cup \mathcal{G}_{32}$.
\end{theorem}
\begin{proof}
  One direction easily follows from Proposition \ref{easyDirection2}. For the other direction, we proceed as follows: Let $G$ be a connected triangle-free EFC-graph having an induced odd cycle $C$ such that ${c(G\backslash C)= 2}$ and let $H_1$ and $H_2$ be two components of $c(G\backslash C)$. If $C$ is a cycle of length seven, then $G \in \mathcal{G}_{11} \cup  \mathcal{G}_{12}$ by Proposition \ref{c7-c7}. If $C$ is a cycle of length five, then by Proposition \ref{c5-c5}, $G$ is a member of $\mathcal{G}_{21} \cup  \mathcal{G}_{22} \cup  \mathcal{G}_{23} \cup \mathcal{G}_{31} \cup  \mathcal{G}_{32}$. Hence, we are done.
 \hfill\(\qed\)
\end{proof}

\section{Main Theorem and Recognition Algorithm}

In  this section, we first present the main result of this paper and then give an efficient recognition algorithm for nonbipartite triangle-free EFC-graphs.

The following theorem, which provides a complete structural characterization of triangle-free equimatchable graphs, is the main theorem of this paper:

\begin{theorem}\label{MainTheorem}
  A graph $G$ is a connected triangle-free equimatchable graph if and only if one of the following holds:
  \begin{description}
    \item[(i)] $G$ is an equimatchable bipartite graph, which is characterized in \cite{Lesk}.
    \item[(ii)] $G$ is a $C_5$ or a $C_7$.
    \item[(iii)] $G \in \mathcal{F} =  \mathcal{F}_{11} \cup \mathcal{F}_{12} \cup  \mathcal{F}_{21} \cup \mathcal{F}_{22} \cup \mathcal{F}_{3} \cup \mathcal{F}_{4}$.
    \item[(iv)] $G \in \mathcal{G} =  \mathcal{G}_{11} \cup  \mathcal{G}_{12} \cup  \mathcal{G}_{21} \cup  \mathcal{G}_{22} \cup  \mathcal{G}_{23} \cup  \mathcal{G}_{31} \cup  \mathcal{G}_{32}$.
  \end{description}
\end{theorem}

\begin{proof}
One direction follows from Proposition \ref{easyDirection} and Proposition \ref{easyDirection2}. We proceed with the other direction. Let $G$ be a connected triangle-free equimatchable graph. If $G$ is a bipartite equimatchable graph, then (i) holds. Now consider the case where $G$ has an induced odd cycle $C$ of length at least $5$. Hence, by Theorem \ref{odd-hole-free}, $G$ is an EFC-graph. If $G$ has girth at least $5$, then (ii) holds by \cite{Frendrup}. Otherwise, by Corollary \ref{TwoComponents}, $G \backslash C$ has at most two components. If $G \backslash C$ is connected, then (iii) holds by Theorem \ref{main1}. If $G \backslash C$ has two components, then (iv) holds by Theorem \ref{main2}.
\hfill\(\qed\)
\end{proof}

The recognition problem of triangle-free equimatchable graphs is clearly polynomial since each one of these two properties can be tested in polynomial time. Equimatchable graphs can be recognized in time
$\cal{O}$ $({m \cdot \overline{m}})$ (see \cite{DemangeEkim}), where $m$ (resp. $\overline{m}$) is the number of edges (resp. non-edges) of the graph. Triangle-freeness can be recognized in time $\cal{O}$ $(m^{\frac{2 \omega}{\omega+1}}=$ $\cal{O}$ $(m^{1.407})$ (see \cite{Alon1997}), where $w$ is the exponent of the matrix multiplication complexity (the best known exponent is $\omega \approx 2.37286$ (see \cite{LeGall})). Hence, the currently known overall complexity of the recognition of triangle-free equimatchable graphs is $\cal{O}$ $(m \cdot(\overline{m} + m^{0.407}))$.

We will now show that for non-bipartite graphs, our characterization yields a linear time recognition algorithm.

\alglanguage{pseudocode}

\begin{algorithm}[H]
\caption{Triangle-free equimatchable graph recognition for non-bipartite graphs}\label{alg:recognition}
\begin{algorithmic}[1]
\Require{A non-bipartite graph $G$.}
\Statex

\State Compute a twin-free graph $H$ and multiplicities $n_1, \ldots, n_k$ such that
$G=H(n_1, \ldots, n_k)$. \label{step:TwinElimination}

\If {$H$ is isomorphic to neither $G^*$ nor to a relevant subgraph of it}  \label{step:Isomorphism}
\State \Return \textbf{false}

\Else
\State \Return \textbf{true} if and only if $n_1, \ldots, n_k$ matches any of the multiplicity patterns in the definitions of $\cal{F}$, $\cal{G}$, $C_5$ or $C_7$ \label{step:multiplicity}

\EndIf

\end{algorithmic}
\end{algorithm}


\begin{corollary}\label{cor:algo}
Given a non-bipartite graph $G$, Algorithm \ref{alg:recognition} can recognize whether $G$ is equimatchable and triangle-free in linear time.
\end{corollary}
\begin{proof}
The correctness of Algorithm \ref{alg:recognition} is a direct consequence of Theorem \ref{MainTheorem}. For every graph $G$, there is a twin-free graph $H$ and a  vector $(n_1, \ldots, n_k)$ of vertex multiplicities such that $G=H(n_1, \ldots, n_k)$. The graph $H$ and the vector $(n_1, \ldots, n_k)$ can be computed from $G$ in linear time by first computing the modular decomposition of $G$ (see \cite{ChristopheModular}) and then looking for leaves of the modular decomposition tree that are independent sets. Therefore, step \ref{step:TwinElimination} can be performed in linear time.

We now note that some entries of the multiplicity vectors allowed for the graph families can be zero. In this case, $H$ is not isomorphic to $G^*$ but to an induced subgraph of it with those specific vertices removed. We refer to those graphs as \emph{relevant subgraphs} in the algorithm. Observe that $C_5$ and $C_7$ are also one of the relevant subgraphs since $C_{5}=G^*(1,1,1,1,1,0,0,0,0,0,0)$ and $C_{7}=G^*(1,1,1,0,1,1,1,1,0,0,0)$.

As for step \ref{step:Isomorphism}, it takes constant time to decide whether an isomorphism exists: if $H$ has more than $11$ vertices, it is isomorphic to neither  $G^*$  nor to a subgraph of them; otherwise, $H$ has to be compared to each one of these graphs and their relevant subgraphs, where each comparison takes constant time. Finally, step \ref{step:multiplicity} takes constant time.

We conclude that the running time of Algorithm \ref{alg:recognition} is dominated by the running time of step \ref{step:TwinElimination}, which can be performed in linear time.
 \hfill\(\qed\)
\end{proof}

\section{Conclusion and Open Questions}
Frendrup et al. \cite{Frendrup} provided a characterization of equimatchable graphs with girth at least 5, while Akbari et al. \cite{AkbariRegular} characterized regular triangle-free equimatchable graphs. In this paper, we extend both of these results by providing a complete structural characterization of triangle-free equimatchable graphs in terms of graph families. For non-bipartite graphs, our characterization implies a linear time recognition algorithm, which improves the known complexity of recognizing triangle-free equimatchable graphs.

Bipartite equimatchable graphs have been characterized by Lesk et al. \cite{Lesk} (Theorem \ref{lesk-thm-bipartiteEqm}); hence, we focus on non-bipartite triangle-free equimatchable graphs. Since it has been proved by Dibek et al. \cite{Dibek} (Theorem \ref{C-odd-free}) that equimatchable graphs do not have induced odd cycles of length at least nine, the only possible induced odd cycles in triangle-free equimatchable graphs are cycles of length five and seven. We then study such graphs by taking into account the fact that equimatchable graphs with an induced odd cycle of length at least five has to be factor-critical (Theorem \ref{odd-hole-free}). We first prove that if there is an induced cycle of length five or seven in such a graph, then every component of the subgraph induced by all vertices that are not on this cycle is isomorphic to a complete bipartite graph with equal partition sizes (Theorem \ref{KnnProof}). We then derive the structure of the graph induced by the odd cycle and the endpoints of an edge outside the cycle (Theorem \ref{C-OneEdge-Structure}) as well as the one induced by the odd cycle and two edges from different components of the subgraph outside the odd cycle (Theorem \ref{C-TwoEdge-Structure}). As a corollary, we show that the subgraph induced by the vertices that are not on the odd cycle has at most two components (Corollary \ref{TwoComponents}). We then analyze the case with one and two components in Section \ref{oneComp} and Section \ref{twoComp}, respectively. In particular, we prove that each such graph can be obtained from the twin-free graph in Figure \ref{fig:MainTheorem} or its relevant subgraphs. More precisely, we provide in Theorem \ref{main1} and Theorem \ref{main2}  the structure of triangle-free equimatchable graphs where the removal of an induced odd cycle leaves one component and two components, respectively. In Section 5, we present the main theorem of the paper (Theorem \ref{MainTheorem}). Building on this structural characterization, we present in Corollary \ref{cor:algo} a linear time algorithm which, given a non-bipartite graph, recognizes whether the graph is equimatchable and triangle-free.

An interesting open question is to obtain a similar efficient recognition algorithm for bipartite equimatchable graphs. Together with the algorithm in this paper, such an algorithm would also imply an efficient recognition algorithm for triangle-free equimatchable graphs in general. The characterization given by Lesk et al. \cite{Lesk} (Theorem \ref{lesk-thm-bipartiteEqm}) is not structural and does not lead to an efficient recognition algorithm. The characterization by Frendrup et al. \cite{Frendrup} (Theorem \ref{frendrup}) is indeed structural; however, it is only for bipartite equimatchable graphs with girth at least six. Extending this characterization to the general case of bipartite equimatchable graphs and obtaining graph families leading to an efficient recognition algorithm is an interesting research direction.

\bibliography{REFERENCES}

\end{document}